\documentclass{amsart}

\usepackage{mathtools}
\usepackage{nicefrac}
 \usepackage{amssymb}
  \usepackage{amsmath}
\usepackage{color}
\usepackage{tikz}
\usetikzlibrary{arrows,shapes,positioning}
\usetikzlibrary{decorations.markings}
\tikzstyle arrowstyle=[scale=1]
\tikzstyle directed=[postaction={decorate,decoration={markings,
    mark=at position .65 with {\arrow[arrowstyle]{stealth}}}}]
\tikzstyle reverse directed=[postaction={decorate,decoration={markings,
    mark=at position .65 with {\arrowreversed[arrowstyle]{stealth};}}}]
\newtheorem{teo}{Theorem}[section]
\newtheorem{lem}[teo]{Lemma}
\newtheorem{co}[teo]{Corollary}

\newtheorem{re}[teo]{Remark}

\newtheorem{de}[teo]{Definition}

\newtheorem{ex}[teo]{Example}

\newcommand{\X}{\mathcal{X}}

\renewcommand{\v}{\mathrm{v}}
\renewcommand{\u}{\mathrm{u}}
\newcommand{\V}{\mathrm{V}}
\newcommand{\T}{\mathrm{T}}

\newcommand{\e}{\mathrm{e}}
\newcommand{\E}{\mathrm{E}}

\begin{document}

\title[{The first eigenvalue of the $p-$Laplacian on quantum graphs}]{The first eigenvalue of the $p-$Laplacian on quantum graphs}
\author[L. M. Del Pezzo and J. D. Rossi]
{Leandro M. Del Pezzo and Julio D. Rossi}

\address{Leandro M. Del Pezzo and Julio D. Rossi
\hfill\break\indent
CONICET and Departamento  de Matem{\'a}tica, FCEyN,
Universidad de Buenos Aires,
\hfill\break\indent Pabellon I, Ciudad Universitaria (1428),
Buenos Aires, Argentina.}

\email{{\tt ldpezzo@dm.uba.ar,
jrossi@dm.uba.ar
}}

\thanks{
Leandro M. Del Pezzo was partially supported by UBACyT 20020110300067
and CONICET PIP 5478/1438  (Argentina) 
and Julio D. Rossi  was partially supported by MTM2011-27998,
(Spain) }

\maketitle

\begin{abstract}
We study the first eigenvalue of the $p-$Laplacian (with $1<p<\infty$) on a quantum graph with Dirichlet or 
Kirchoff boundary conditions on the nodes. We find lower and upper bounds for this 
eigenvalue when we prescribe the total sum of the lengths of the edges and the number of 
Dirichlet nodes of the graph. Also we find a formula for the shape derivative of the 
first eigenvalue (assuming that it is simple) when we perturb the graph 
by changing the length of an edge. Finally, we study in detail the limit cases $p\to \infty$ 
and $p\to 1$. 
\end{abstract}

\section{Introduction}

A quantum graph is a graph in which we associate a differential law with each edge. This 
differential law models the interaction between the two nodes defining each edge. The 
use of quantum graphs (as opposed to more elementary graph models, such as simple 
unweighted or weighted graphs) opens up the possibility of modeling the interactions 
between agents identified by the graph's vertices in a far more detailed manner than 
with standard graphs. Quantum graphs are now widely used in physics, chemistry and 
engineering (nanotechnology) problems, but can also be used, in principle, in the 
analysis of complex phenomena taking place on large complex networks, including social 
and biological networks. Such graphs are characterized by highly skewed degree 
distributions, small diameter and high clustering coefficients, and they have 
topological and spectral properties that are quite different from those of the highly 
regular graphs, or lattices arising in physics and chemistry applications. 
Quantum graphs are also used to model
thin tubular structures, so-called graph-like spaces, they are their natural 
limits, when the radius of a graph-like space tends to zero. On both, the graph-like 
spaces and the metric graph, we can naturally define Laplace-like differential 
operators. See \cite{Liviu,BK,K,olaf}. 

Among properties that are relevant in the study of quantum graphs is the study of the 
spectrum of the associated differential operator. In particular, the so-called spectral 
gap (this concerns bounds for the first nontrivial eigenvalue for the Laplacian with 
Neumann boundary conditions) has physical relevance and was extensively studied in 
recent years. See, for example, \cite{K,Ku,KN2,Ku22} and references therein.

In this paper we are interested in the eigenvalue problem that naturally arises when we 
consider the $p-$Laplacian, $(|u'|^{p-2} u')'$, as the differential law on each side of 
the graph together with Dirichlet boundary conditions on a subset of nodes of the graph 
and pure transmission (known as Kirchoff boundary conditions, \cite{cita2}) in the rest 
of the nodes. To be concrete, given $1<p<\infty$, we deal with the following problem: 
in a finite metric graph $\Gamma$ we consider a set of nodes $V_D$ and look for the 
minimization problem
\begin{equation}
	\label{eq:autovalor.intro}
	\lambda_{1,p}(\Gamma,\V_D)= \inf
	\left\{\dfrac{\displaystyle\int_{\Gamma} |u^{\prime}(x)|^p\, dx}
	{\displaystyle\int_{\Gamma} |u(x)|^p\, dx}\colon 
	u\in\mathcal{X}(\Gamma,\V_D), u\neq0\right\},
\end{equation}
where $\mathcal{X}(\Gamma,\V_D)\coloneqq 
\{v\in W^{1,p}(\Gamma)\colon v\mbox{ is continuous in }\Gamma, \, v=0 \mbox{ on } 
\V_D\}$.

There is a minimizer, see Section \ref{sect-first}, that is a nontrivial weak solution 
to 

\begin{equation}\label{eq:vc.intro}
\left\{	\begin{array}{ll}
	\displaystyle
		- (|u^\prime |^{p-2}
	u^\prime)^{\prime} (x) = \lambda_{1,p}(\Gamma,\V_D) |u|^{p-2} u (x) \quad &
	\mbox{ on the edges of } \Gamma,\\[6pt]
		u(\v)=0\quad & \forall \v\in \V_D,\\[6pt]
		\displaystyle\sum_{\e\in \E_{\v}(\Gamma)}  
		\left|\dfrac{\partial u}{\partial x_{\e}}(\v)\right|^{p-2}
		\dfrac{\partial u}{\partial x_{\e}}(\v)=0\quad & 
		\forall \v\in \V(\Gamma)\setminus \V_D.
	\end{array} \right.
\end{equation}

Our main results for this eigenvalue problem can be summarized as follows (we refer to 
the corresponding sections for precise statements):

\begin{itemize}

\item We show that there is a first eigenvalue with an associated nonnegative 
eigenfunction, that is, the infimum in \eqref{eq:autovalor.intro} is attained at a nonnegative function.
We provide examples that show that $\lambda_{1,p}(\Gamma,\V_D)$ can be a multiple eigenvalue or a simple eigenvalue depending on the graph.

\medskip

\item We find a sharp lower bound for the first eigenvalue that depends only on the 
total sum of the lengths of the edges of the graph, $\ell(\Gamma)$, namely
$$
		\lambda_{1,p}(\Gamma,\V_D)\ge C(p)
		\left(\dfrac{1}{\ell(\Gamma)}\right)^p ,
$$
here the constant $C(p)$ is explicit and depends only on $p$.

\medskip

\item We find a sharp upper bound for the first eigenvalue depending on the total sum of 
the lengths of the edges, $\ell(\Gamma)$, and the number of edges of the graph, 
$\mathrm{card}(E(\Gamma))$,
\[
		\lambda_{1,p}(\Gamma,\V_D)\le C(p)
		\left(\dfrac{\mathrm{card}(E(\Gamma))}{\ell(\Gamma)}\right)^p,
	\]
again the constant $C(p)$ is explicit and depends only on $p$.

\medskip

\item Under the assumption that the first eigenvalue is simple, we find a formula for 
its shape derivative when we perturb the graph by changing the length of an edge. In the case of a multiple eigenvalue, we provide 
examples that show that the first eigenvalue is not differentiable with respect to the 
lengths of the edges of the graph (but it is Lipschitz).

\medskip

\item We study the limit cases $p\to \infty$ and $p\to 1$. For $p=\infty$ we find a 
geometric characterization of the first eigenvalue and for $p=1$ we prove that there 
exist the analogous of Cheeger sets in quantum graphs.
\end{itemize}

Note that without a bound on the total length of the graph the first eigenvalue is 
unbounded from above and from below the optimal bound is zero
and without a bound on the number of Dirichlet nodes it is not bounded above even 
if we prescribe the total length. Therefore our 
results are also sharp in this sense. 
Also remark that our results are new even for the linear case $p=2$.

\medskip

Let us end this introduction with a brief discussion on ideas and techniques used in the 
proofs as well as a description of the previous bibliography.

Existence of eigenfunctions can be easily obtained from a compactness argument as for the 
usual $p-$Laplacian in a bounded domain of ${\mathbb{R}}^N$, see \cite{GP1}. However, 
in contrast to what happens in the usual case of a bounded domain, see \cite{anane}, the 
first eigenvalue is not simple, we show examples of this phenomena. 

Eigenvalues on quantum graphs are by now a classical subject with an increasing number of 
recent references, we quote \cite{bon,cita,Ku,Ku22}. The literature on eigenfunctions of 
the $p-$Laplacian, also called 
$p-$trigonometric functions, is now quite extensive: we refer in particular to 
\cite{11,12,13} and references therein. 

The upper and lower bounds comes from test functions arguments together with some 
analysis of the possible configurations of the graphs. 

For the shape derivative when we modify the length of one edge we borrow ideas from 
\cite{jorge-pepe}.

Concerning the limit as $p\to \infty$ for the eigenvalue problem of the $p-$Lapla-cian 
in the usual PDE case we refer to \cite{BD1,BD2,Juu,JLM2}. To obtain this limit the main 
point is to use adequate test functions to obtain bounds that are uniform in $p$ in order 
to gain compactness on a sequence of eigenfunctions. 

Finally, for $p=1$ we refer to \cite{CCN,FMP,parini}. In this limit problem the natural 
space that appear is that of bounded variation functions, see \cite{ambrosio}. Remark 
that when considering bounded variation functions we loose continuity.

\medskip

The paper is organized as follows: in Section \ref{sect.prelim} we collect some 
preliminaries; in Section 
\ref{sect-first} we deal with the first eigenvalue on a quantum graph and prove its 
upper and lower bounds; 
in Section \ref{sect-shape} we perform a shape derivative approach of the first 
eigenvalue showing that it is differentiable when we change the length of one edge and 
providing an explicit formula for this derivative;  
in Section \ref{sect-infty} we study the limit as $p \to \infty$ of the first 
eigenvalue while in the final section, Section \ref{sect-1} we look for the limit as 
$p\to 1$.

\section{Preliminaries.} \label{sect.prelim}
\subsection{Quantum Graphs} We collect here some basic knowledge about quantum graphs, 
see for instance 
\cite{BK} and references therein.

A graph $\Gamma$ consists of a finite or countable infinite set 
of vertices $\V(\Gamma)=\{\v_i\}$ and a set of edges $\E(\Gamma)=\{\e_j\}$ connecting the 
vertices. A graph $\Gamma$ is said a finite graph if the number of edges 
and the number of vertices are finite.

Two vertices $\u$ and $\v$ are called adjacent (denoted $\u\sim \v$) if there is an 
edge connecting them. An edge and a vertex on 
that edge are called incident. We will denote $\v\in \e$ when $\e$ and $\v$ are incident.
We define $\E_{\v}(\Gamma)$ as the set of all edges incident to $\v.$ The degree $d_{\v}(\Gamma)$ of a vertex 
$\V(\Gamma)$ is the number of edges that incident to it, where a loop 
(an edge that connects a vertex to itself) is counted twice.

We will say that $\v$ is a terminal vertex if 
there exists an unique vertex $\u\in\V(\Gamma)$ such that $\u\sim\v.$
Let us denote by $\T(\Gamma)$ the set of all terminal vertices.

A walk is a sequence of edges in which the end of each edge (except the last) is 
the beginning of the next. A trail is a walk in which no edge is repeated.
A path is a trail in which no vertex is repeated. A graph $\Gamma$ is said 
connected if a path exists between every pair of vertices, that is a graph which is 
connected in the sense of a topological space.

A graph $\Gamma$ is called a directed graph if each of its edges is assigned a direction.
In the remainder of the section, $\Gamma$ is a directed graph.

Each edge $\e$ can be identified with an ordered pair 
$(\v_{\e},\u_{\e})$ of vertices.The vertices $\v_{\e}$ and $\u_{\e}$ are the 
initial and terminal 
vertex of $\e.$ The edge $\hat{\e}$ is called the
reversal of the edge $\e$ if $\v_{\hat{\e}}=\u_{\e}$ and 
$\u_{\hat{\e}}=\v_{\e}.$ We define 
$$
	\widehat\E(\Gamma)\coloneqq\{\hat e\colon e\in\E(\Gamma)\}.
$$ 

The edge $\e$ is called outgoing (incoming) at a vertex $\v$ if 
$\v$ is the initial (terminal) vertex of $\e.$ The number of outgoing (incoming)
edges at a vertex $\v$ is called outgoing (incoming) degree and denoted $d_{\v}^o(\Gamma)$
($d_{\v}^i(\Gamma)$). Observe that $d_{\v}(\Gamma)=d_{\v}^o(\Gamma)
+d_{\v}^i(\Gamma).$ 

\begin{de}[See Definition 1.2.3 in \cite{BK}] 
A graph $\Gamma$ is said to be a metric graph, if
\begin{enumerate}
	\item each edge $\e$ is assigned a positive length $\ell_{\e}\in(0,+\infty];$
	\item the lengths of the edges that are reversals of each other are assumed
			to be equal, that is $\ell_{\e}=\ell_{\hat{\e}};$
	\item a coordinate $x_{\e}\in I_{\e}=[0,\ell_{\e}]$ 
			increasing in the direction of the edge
			is assigned on each edge;
	\item the relation $x_{\hat{\e}}=\ell_{\e}-x_{\e}$ holds between the coordinates on
			mutually reserved edges. 
\end{enumerate}
\end{de}

A finite metric graph whose edges all have finite 
lengths will be called compact. 
If a sequence of edges $\{\e_j\}_{j=1}^n$ forms a path, its length is 
defined as $\sum_{j=1}^n\ell_{\e_j}.$ For two vertices $\v$ and $\u,$ the distance $d(\v,\u)$ is defined as the 
minimal length of the path connected them. A compact 
metric graph $\Gamma$ becomes a metric measure space by defining the 
distance $d(x,y)$ of two points $x$ and $y$ of the graph (that are not necessarily vertices) to be the short path on $\Gamma$ connected these
points, that is
\[
	d(x,y) \coloneqq \inf\left\{
	\int_0^1 |\gamma' (t)| \, dt\colon
	\gamma\colon [0,1] \to \Gamma \mbox{ Lipschitz}, 
	\ \gamma(0) =x, \  \gamma(1)=y 
	\right\}.
\]
  
The length of a metric graph 
(denoted $\ell(\Gamma)$) is the sum of the length of all edges.

A function $u$ on a metric graph $\Gamma$ is a collection of functions $u_{\e}$ 
defined on
$(0,\ell_{\e})$ for all $\e\in \E(\Gamma),$ 
not just at the vertices as in discrete models.  

Let $1\le p\le \infty.$ We say that $u$ belongs to  $L^p(\Gamma)$ if 
$u_{\e}$ belongs to $L^p(0,\ell_{\e})$ for all $\e\in \E(\Gamma)$ and
\[
	\|u\|_{L^{p} (\Gamma)}^p\coloneqq\sum_{\e\in \E(\Gamma)}
	\|u_{\e}\|_{L^{p}(0,\ell_{\e})}^p<\infty.
\]
 The Sobolev space $W^{1,p}(\Gamma)$ is defined as 
the space of continuous functions $u$ on $\Gamma$ such  that 
$u_{\e}\in W^{1,p}(I_{\e})$ 
for all $\e\in \E(\Gamma)$ and
\[
	\|u\|_{W^{1,p}(\Gamma)}^p\coloneqq \sum_{\e\in \E(\Gamma)}
	\|u_{\e}\|_{L ^p(0,\ell_{\e})}^p+\|u^\prime_{\e}\|_{L ^p(0,\ell_{\e})}^p<\infty.
\] 
Observe that the continuity condition in the definition of $W^{1,p}(\Gamma)$ 
means that for 
each $\v\in \V(\Gamma),$ the function on all edges 
$\e\in \E_{\v}(\Gamma)$  assume the 
same value at $\v.$

The space $W^{1,p}(\Gamma)$ is a Banach space for $1 \le p \le\infty$. 
It is reflexive for $1 < p < \infty$ and separable for $1 \le p < \infty.$

\begin{teo}
	\label{teo:inccomp} Let $\Gamma$ be a compact graph and $1<p<\infty$. 
	The injection $W^{1,p}(\Gamma)\subset L^q(\Gamma)$
	is compact for all $1\le q\le\infty.$
\end{teo}

A quantum graph is a metric graph $\Gamma$ equipped with a 
differential operator 
$\mathcal{H},$ accompanied by a vertex conditions.
In this work, we will consider 
\[
	\mathcal{H}(u)(x)\coloneqq -\Delta_p u(x)= -
	(|u^{\prime}(x)|^{p-2} u^{\prime}(x))^{\prime}.
\]
Given $\V_D$ a  non empty subset of $\V(\Gamma),$  our vertex 
conditions are the following
\begin{equation}\label{eq:vc}
	\begin{cases}
		u(x)\mbox{ is continuous in }\Gamma,\\
		u(\v)=0\quad\forall \v\in \V_D,\\
		\displaystyle\sum_{\e\in \E_{\v}(\Gamma)}  \left|\dfrac{\partial u}{\partial x_{\e}}(\v)\right|^{p-2}
		\dfrac{\partial u}{\partial x_{\e}}(\v)=0\quad 
		\forall \v\in \V(\Gamma)\setminus \V_D,
	\end{cases}
\end{equation}
where the derivatives are assumed to be taken in the direction 
away from the vertex. 

Throughout this work, $ \int_{\Gamma} u(x)\, dx$ denotes
$ \sum_{\e\in \E(\Gamma)} \int_{0}^{\ell_{\e}} u_{\e}(x)\, dx$.

\subsection{Eigenvalues of the $p-$Laplacian in $\mathbb{R}$}
Here we present a brief review concerning eigenvalues 
of the 1-dimensional $p-$Laplacian. For a more elaborate
treatment we refer the reader to \cite{LE}.

Let $p\in(1,+\infty).$ 
Given $L>0,$ all eigenvalues $\lambda$ of the Dirichlet problem
\[
	\begin{cases}
		-(|u^{\prime}|^{p-2} u^{\prime})^{\prime}=\lambda
		|u|^{p-2}u &\text{ in } (0,L),\\
		u(0)=u(L)=0,
	\end{cases}
\]
are of the form
\[
	\lambda_{n,p}=\left(\dfrac{n\pi_{p}}{L}\right)^p
	\dfrac{p}{p^\prime}\quad\forall n\in\mathbb{N}
\]
with corresponding eigenfunctions
\[
	u_{n}(x)=\dfrac{\alpha L}{n\pi_p}\sin_p\left(\dfrac{n\pi_p}{L}x\right),
	\quad \alpha\in\mathbb{R}\setminus\{0\}
\]
where 
$
	\pi_p=\frac{2\pi}{p\sin(\nicefrac{\pi}p)}$,
$\nicefrac{1}{p}+\nicefrac{1}{p^\prime}=1,$ and $\sin_p$ is the $p-$sine 
function. 

Then the first Dirichlet eigenvalue is 
\begin{equation}
	\label{eq:primerautov}
	\lambda_{1,p}=\left(\dfrac{\pi_{p}}{L}\right)^p
	\dfrac{p}{p^\prime},
\end{equation}
and has a positive eigenfunction (any other eigenvalue has eigenfunctions that
change sign). 

\begin{re}
	Observe that $\{\lambda_{n,p}\}$ coincides with the Dirichlet eigenvalues 
	of the Laplacian when $p=2$.
\end{re}

\section{The first eigenvalue on a quantum graph.} \label{sect-first}

Let $\Gamma$ be a compact connected quantum graph and $\V_D$
be a non-empty subset of $\V(\Gamma).$ We say that  the value 
$\lambda\in\mathbb{R}$ is an eigenvalue of  the $p-$Laplacian
if there exists non trivial function
$u\in \mathcal{X}(\Gamma,\V_D)\coloneqq 
\{v\in W^{1,p}(\Gamma)\colon v=0 \mbox{ on } \V_D\}$ 
such that
\[
	\int_{\Gamma}|u^\prime(x)|^{p-2}
	u^\prime(x)w^\prime(x)\, dx=\lambda 
	\int_{\Gamma}|u(x)|^{p-2}
	u(x)w(x)\, dx
\]
for all $w\in\mathcal{X}.$ In which case, $u$ is called an 
eigenfunction associated to $\lambda.$

Recall from the introduction that the first eigenvalue of the $p-$Laplacian is given by
\begin{equation}
	\label{eq:autovalor}
	\lambda_{1,p}(\Gamma,\V_D)= \inf
	\left\{\dfrac{\displaystyle\int_{\Gamma} |u^{\prime}(x)|^p\, dx}
	{\displaystyle\int_{\Gamma} |u(x)|^p\, dx}\colon 
	u\in\mathcal{X}(\Gamma,\V_D), u\neq0\right\}.
\end{equation}

By a standard compactness argument, it follows that there exists
an eigenfunction associated to $\lambda_{1,p}(\Gamma,\V_D)$. Note that
when $V_D\neq \emptyset$ the norm in $W^{1,p}(\Gamma)$ is equivalent to 
$(\int_{\Gamma} |u^{\prime}|^p)^{1/p} = ( \sum_{\e\in \E(\Gamma)}
\|u^\prime_{\e}\|_{L ^p(0,\ell_{\e})}^p)^{1/p}$.

\begin{teo}
Let $\Gamma$ be a compact connected quantum graph, $\V_D$
be a non-empty subset of $\V(\Gamma)$ and $p\in(1,+\infty)$. 
Then there exists a non-negative $u_0\in\mathcal{X}(\Gamma,\V_D)$ such that 
\[
	\lambda_{1,p}(\Gamma,\V_D)= \dfrac{\displaystyle\int_{\Gamma} |u_0 ^{\prime}(x)|^p\, dx}
	{\displaystyle\int_{\Gamma} |u_0(x)|^p\, dx}.
\]
Moreover, $u_0$ is an eigenfunction associated to $\lambda_{1,p}(\Gamma,V_D).$
\end{teo}
\begin{proof}
Let $\{u_n\}_{n\in\mathbb{N}}\subset\mathcal{X}(\Gamma,V_D)$ 
be a minimizing sequence for $\lambda_{1,p}(\Gamma,\V_D),$ that is,
\[
	\lambda_{1,p}(\Gamma,\V_D)=\lim_{n\to\infty}\int_{\Gamma} |u_n^{\prime}(x)|^p\, dx, \quad
	\int_{\Gamma} |u_n(x)|^p\, dx=1\quad \forall n\in\mathbb{N}.
\] 
Note that we can assume that $u_n\ge0.$
Then, there exists $C>0$ such that $\|u_n\|_{W^{1,p}(\Gamma)}\le C$ for all $n\in\mathbb{N}.$
Therefore, using that $\X(\Gamma,\V_D)$ is a reflexive space and Theorem \ref{teo:inccomp},  there exist $u_0\in\mathcal{X}(\Gamma,V_D)$ and a subsequence that will still call $\{u_n\}_{n\in\mathbb{N}}$
such that
\begin{align}
	u_n\rightharpoonup u_0, &\mbox{ weakly in } 
	\X(\Gamma,\V_D),\label{convdebil}\\
	u_n\to u_0, &\mbox{ strongly in }L^p(\Gamma).\label{convfuerte}
\end{align}
As $\|u_n\|_{L^p{(\Gamma)}}=1$ for all $n\in\mathbb{N},$ by \eqref{convfuerte}, we have that $\|u_0\|_{L^p(\Gamma)}=1.$
Then $u_0\neq 0.$ 

On the other hand, by \eqref{convdebil}, 
\[
	\lambda_{1,p}(\Gamma,\V_D)=\lim_{n\to\infty}\int_{\Gamma} |u_n^{\prime}(x)|^p\, dx\ge \int_{\Gamma} |u^{\prime}_0(x)|^p\, dx. 
\]
Then, by \eqref{eq:autovalor}, we get
\[
	\lambda_{1,p}(\Gamma,\V_D)=\int_{\Gamma} |u^{\prime}_0(x)|^p\, dx. 
\]
Finally, it is clear that $u_0$ is an eigenfunction of the $p-$Laplacian associated to 
$\lambda_{1,p}(\Gamma,\V_D).$
\qed
\end{proof}

\begin{re}
	Note that, if $\V_D\subset\V_{D}^\prime\subset\V(\Gamma)$ then 
	$\lambda_1(\Gamma,\V_D)\le \lambda_1(\Gamma,\V_D^\prime),$ due to
	$\X(\Gamma,V_D^\prime)\subset\X(\Gamma,V_D).$	
\end{re}

Our next result shows that the first eigenvalue is simple if the Dirichlet
vertices are terminal vertices.

\begin{teo}\label{teo:simplicidad}
	Let $\Gamma$ be a compact connected quantum graph such that 
	$\T(\Gamma)\neq\emptyset$, and $p\in(1,+\infty).$ 
	If $\V_D\subseteq\T(\Gamma)$ is non-empty
	then the eigenfunctions associated to $\lambda_{1,p}(\Gamma,\V_D)$  
	do not change sign and, in addition, $\lambda_{1,p}(\Gamma,\V_D)$ is simple. 
	Here $\mathrm{card}(\V(\Gamma))$ is the cardinal number of $\V(\Gamma).$
\end{teo}
\begin{proof}
	Let $u$ be an eigenfunction associated to $\lambda_{1,p}(\Gamma,\V_D).$
	We have that $|u|$ is also a minimizer of \eqref{eq:autovalor}. Then,
	without loss of generality, we can assume that $u\ge0$ in $\Gamma.$
	
	Let $\v\in\V_D$ and $\u\in \V(D)$ such that $\v\sim \u$ and
	$u\neq0$ in $I_{e_0}$ where $e_0\in\E(\Gamma)$ and $\v,\u\in e_0.$
	Then, by the maximum principle (see \cite{vazquez}), we have that
	$u>0$ in $(0,\ell_e).$ Moreover if $u(\u)=0$, 
	by Hopf's lemma, $u^\prime(\u)>0,$ and this contradicts the Kirchhoff conditions 
	at $\u.$ Hence $u(\u)>0.$ 
	Then $u>0$ in $(0,\ell_e)$ for all $e\in E_{\u}(\Gamma).$ 
	We continue in this fashion
	obtaining $u>0$ in $\Gamma$. Once we have that every eigenfunction does 
	not change sign we get
	simplicity for $\lambda_{1,p}(\Gamma,\V_D)$ arguing as in \cite{11}.
	\qed
\end{proof}

\begin{re}\label{re:nosimple}
	In general, the first eigenvalue is not simple.
	For example,  let $\Gamma$ be a simple graph with 3 vertices
	and 2 edges, that is $\V(\Gamma)=\{\v_1,\v_2, v_3\}$ and 
	$\E(\Gamma)=\{[\v_1,\v_2], [\v_2,\v_3]\}.$ Let $\V_D=\{\v_1,\v_2,\v_3\}.$
	
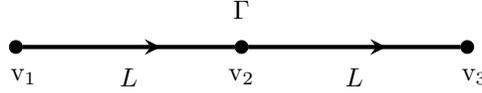
	\begin{center}
	\begin{tikzpicture}
			 \node (a0) at (3,.5) {$\Gamma$};
  		    \node  [fill, circle,draw, scale=.5]  (a) at (0,0) {};
  			\node  (a1) at (1.5,-0.4)  { $L$ };
  			\node  (a1) at (4.5,-0.4)  { $L$ };
  			\node  (c) at (1.5,0)  {};
  			\node  [fill, circle,draw, scale=.5]  (b) at (3,0)  {};
  			\node  [fill, circle,draw, scale=.5]  (d) at (6,0)  {};
  			\node  (a2) at (0.1,-0.4)  { $\v_1$ };
  			\node  (b1) at (3,-0.4)  { $\v_2$ };
  			\node  (d1) at (6.1,-0.4)  { $\v_3$ };
  			\draw[directed,ultra thick] (a) -- (b);
  			\draw[directed,ultra thick] (b) -- (d);	
		\end{tikzpicture}
	\end{center}
	Then $\lambda_{1,p}(\Gamma,V_{D})=
	\left(\dfrac{\pi_{p}}{L}\right)^p\dfrac{p}{p^\prime}$ and
	\begin{align*}
		&u(x)=\begin{cases}
			\dfrac{L}{\pi_p}\sin_p\left(\dfrac{\pi_p}{L}t\right), 
			&\text{ if } x\in I_{[\v_1,\v_2]}=[0,L],\\
			0 &\text{otherwise},
		\end{cases}\\
		&v(x)=\begin{cases}
			\dfrac{L}{\pi_p}\sin_p\left(\dfrac{\pi_p}{L}t\right), 
			&\text{ if } x\in I_{[\v_2,\v_3]}=[0,L],\\
			0 &\text{otherwise},
		\end{cases}
	\end{align*}
	are two linearly independent eigenfunctions associated to 
	$\lambda_{1,p}(\Gamma,V_D).$ The reason for this lack of simplicity is 
	that the vertex $\v_2$ can be understood as a node that disconnects $\Gamma.$	
\end{re}

Now, we give a lower bound for the first eigenvalue of the $p-$Laplacian 
which does not depend on $\V(\Gamma),$ $\E(\Gamma)$ and $\V_D$. 
For the proof of the next theorem
we follow the ideas of \cite{KN}.
\begin{teo}
	\label{teo:cotainf} Let $\Gamma$ be a connected compact metric graph,
	$\V_D$ be a non-empty subset of $\V(\Gamma)$ and
	$p\in(1,+\infty).$
	Then
	\[
		\lambda_{1,p}(\Gamma,\V_D)\ge 
		\left(\dfrac{\pi_p}{2\ell(\Gamma)}\right)^p\dfrac{p}{p^\prime}.
	\]
\end{teo}
\begin{proof}
Let $\widetilde{\Gamma}$ be a metric graph obtained from $\Gamma$ 
by doubling each edge. Then 
$\E(\widetilde{\Gamma})= \E(\Gamma)\cup\widehat{\E}(\Gamma),$
$\V(\widetilde{\Gamma})= \V(\Gamma),$ and $d_{\v}(\widetilde{\Gamma})$ 
is even for all $\v\in\V(\widetilde{\Gamma}).$

On the other hand, given $u\in\X(\Gamma,\V_D)$ we can define 
$\widetilde{u}\in\X(\widetilde{\Gamma},\V_{D})$ such that
\[
\begin{aligned}
	\widetilde{u}_e(x_{\e})&= u_e(x_{\e}) 
	&\forall x_{\e}\in I_{\e} 
	&&\mbox{ if } \e\in \E(\Gamma)\\
	\widetilde{u}_e(x_{\e})&= u_e(\ell_{\e}-x_{\e}) 
	&\forall x_{\e}\in I_{\e}&& \mbox{ otherwise }.
\end{aligned}
\]
Moreover
\[
	\int_{\widetilde{\Gamma}} |\widetilde{u}^{\prime}(x)|^p\, dx
	= 2 \int_{\Gamma} |u^{\prime}(x)|^p\, dx,\quad \mbox{ and }\quad
	\int_{\widetilde{\Gamma}} |\widetilde{u}(x)|^p\, dx
	= 2 \int_{\Gamma} |u(x)|^p\, dx.
\]
Then 
\begin{equation}
	\label{eq:des1}
	\lambda_{1,p}(\widetilde{\Gamma},\V_D)\le \lambda_{1,p}(\Gamma,\V_D).
\end{equation}

On the other hand, 	there exists a closed path on $\widetilde\Gamma$
coming along every edge in $\widetilde\Gamma$ precisely one time, due to 
$d_{\v}(\widetilde\Gamma)$ is even for all $\v\in\V(\widetilde\Gamma),$
see \cite{Euler,HC}. We identify this path with a loop $\mathfrak{L}$ 
on a vertex $\v_0\in\V_D$ of length less than or equal to $2\ell(\Gamma).$
Observe that $\mathfrak{L}$ is a metric graph,
\begin{equation}
	\label{eq:des2}
	\ell(\mathfrak{L})\le2\ell(\Gamma)\quad\mbox{ and }\quad 
	\lambda_{1,p}(\mathfrak{L},,\{\v_0\})
	\le \lambda_{1,p}(\widetilde\Gamma,\V_D).
\end{equation}
Moreover, 
\begin{align*}
	\lambda_{1,p}(\mathfrak{L},\{\v_0\})&=\inf\left\{
		\dfrac{\displaystyle\int_{\mathfrak{L}} |u^\prime|^p \, dx}
		{\displaystyle\int_{\mathfrak{L}} |u|^p \, dx}\colon
		u\in\X(\mathfrak{L},\{\v_0\}), u\neq0
		\right\}\\	
		&=\inf\left\{
		\dfrac{\displaystyle\int_0^{\ell(\mathfrak{L})} |u^\prime|^p \, dx}
		{\displaystyle\int_0^{\ell(\mathfrak{L})} |u|^p \, dx}\colon
		u\in W^{1,p}_0(0,\ell(\mathfrak{L})), u\neq0
		\right\}\\
		&=\left(\dfrac{\pi_{p}}{\ell(\mathfrak{L})}\right)^p
		\dfrac{p}{p^\prime}	\quad (\mbox{by \eqref{eq:primerautov}}).
\end{align*}
Therefore, by \eqref{eq:des1} and \eqref{eq:des2},
\[
	\lambda_{1,p}(\Gamma,\V_D)\ge\lambda_{1,p}(\widetilde\Gamma,\V_D)
	\ge\lambda_{1,p}(\mathfrak{L},\{v_0\})=
	\left(\dfrac{\pi_{p}}{\ell(\mathfrak{L})}\right)^p
		\dfrac{p}{p^\prime}\ge 
		\left(\dfrac{\pi_{p}}{2\ell(\Gamma)}\right)^p
		\dfrac{p}{p^\prime},
\]
which is the desired conclusion.\qed
\end{proof}

The lower bound given in the above theorem is optimal as the 
following example shows.

\begin{ex}\label{ex:cotainf}
	Let $\Gamma$ be a simple graph with 2 vertices
	 and an edge, that is $\V(\Gamma)=\{\v_1,\v_2\}$ and 
	 $\E(\Gamma)=\{[\v_1,\v_2]\}.$ Let
	 $\V_D=\{\v_1\}.$
	
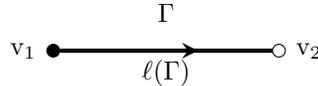
	\begin{center}
	\begin{tikzpicture}
			\node (a0) at (1.5,.5) {$\Gamma$};
  		    \node  [fill, circle,draw, scale=.5]  (a) at (0,0) {};
  			\node  (a1) at (1.5,-0.3)  { $\ell(\Gamma)$ };
  			\node  (c) at (1.5,0)  {};
  			\node  [ circle,draw, scale=.5]  (b) at (3,0)  {};
  			\node  (a2) at (-0.4,0)  { $\v_1$ };
  			\node  (b1) at (3.4,0)  { $\v_2$ };
  			\draw[directed,ultra thick] (a) -- (b);	
		\end{tikzpicture}
	\end{center}
	Then
	\begin{align*}
	\lambda_{1,p}(\Gamma,\V_D)&=\inf\left\{
		\dfrac{\displaystyle\int_{\Gamma} |u^\prime|^p \, dx}
		{\displaystyle\int_{\Gamma} |u|^p \, dx}\colon
		u\in\X(\Gamma,\V_D), u\neq0
		\right\}\\	
		&=\inf\left\{
		\dfrac{\displaystyle\int_0^{\ell(\Gamma)} |u^\prime|^p \, dx}
		{\displaystyle\int_0^{\ell(\Gamma)} |u|^p \, dx}\colon
		u\in W^{1,p}(0,\ell(\Gamma)), u(0)=0, u\neq0
		\right\}\\
		&=\inf\left\{
		\dfrac{\displaystyle\int_0^{2\ell(\Gamma)} |u^\prime|^p \, dx}
		{\displaystyle\int_0^{2\ell(\Gamma)} |u|^p \, dx}\colon
		u\in W^{1,p}_0(0,2\ell(\Gamma)),  u\neq0
		\right\}\\
		&=\left(\dfrac{\pi_{p}}{2\ell(\Gamma)}\right)^p
		\dfrac{p}{p^\prime}.
	\end{align*}
\end{ex}

\begin{ex}\label{ex:cotasup*.2}
	Let $\Gamma$ be a star graph with $n+1$ vertices
	 and $n$ edges, that is $\V(\Gamma)=\{\v_0,\v_1,\dots,\v_n\}$ and 
	 $\E(\Gamma)=\{[\v_1,\v_0], [\v_1,\v_2], \dots, [\v_1,\v_n]\}.$ Let
	 $\V_D=\{\v_1\},$ $\varepsilon>0$ 
	 and $\ell([\v_1,\v_0])=L-(m-1)\varepsilon$
	 and $\ell([\v_1,\v_i])=\varepsilon$
	 for all $i\in\{2,\dots, n\}.$
	 Then $\ell(\Gamma)=L$.
	 
	 
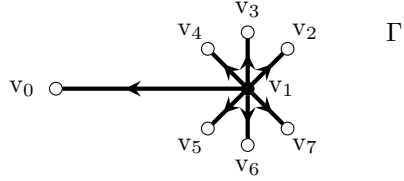
\medskip
	 
	\begin{center}
	\begin{tikzpicture}[scale=1.5]
			\node 	(a0) at (1.3,.5) {$\Gamma$};
  		    \node  [fill,circle,draw, scale=.5]  (a) at (0,0) {};
 
  			\node  (c) at (1.5,0)  {};
  			\node  [ circle,draw, scale=.5]  (b) at (-1.7,0)  {};
  			\node  [ circle,draw, scale=.5]  (c) at 
  			(0.35355339059,0.35355339059)  {};
  			\node  [  circle,draw, scale=.5]  (d) at 
  			(0,.5)  {};
  			\node  [ circle,draw, scale=.5]  (e) at 
  			(-0.35355339059,0.35355339059)  {};
  			\node  [  circle,draw, scale=.5]  (f) at 
  			(0.35355339059,-0.35355339059)  {};
  			\node  [  circle,draw, scale=.5]  (g) at 
  			(-0.35355339059,-0.35355339059)  {};
  			\node  [  circle,draw, scale=.5]  (h) at 
  			(0,-.5)  {};
  			\node  (b0) at (0.3,0)  { $\v_1$ };
  			\node  (b1) at (-2,0)  { $\v_0$ };
  			\node  (b2) at (0,.7)  { $\v_3$ };
  			\node  (b3) at (0,-.7)  { $\v_6$ };
  			\node  (c1) at (0.50710678118,0.50710678118)  { $\v_2$ };
  			\node  (c2) at (-0.50710678118,0.50710678118)  { $\v_4$ };
  			\node  (c3) at (-0.50710678118,-0.50710678118)  { $\v_5$ };
  			\node  (c3) at (0.50710678118,-0.50710678118)  { $\v_7$ };
  			\draw[directed,ultra thick] (a) -- (b);	
  			\draw[directed,ultra thick] (a) -- (c);	
  			\draw[directed,ultra thick] (a) -- (d);	
  			\draw[directed,ultra thick] (a) -- (e);	
  			\draw[directed,ultra thick] (a) -- (f);	
  			\draw[directed,ultra thick] (a) -- (g);	
  			\draw[directed,ultra thick] (a) -- (h);	
		\end{tikzpicture}
	\end{center}
	Then
	\[
		\lambda_{1,p}^{\varepsilon}(\Gamma,V_D)=
		\left(\dfrac{\pi_{p}}{2(L-(n-1)\varepsilon)}\right)^p
		\dfrac{p}{p^\prime}\to\left(\dfrac{\pi_{p}}{2L}\right)^p
		\dfrac{p}{p^\prime}=\left(\dfrac{\pi_{p}}{2L}\right)^p
		\dfrac{p}{p^\prime}
	\]
	as $\varepsilon\to0^+.$
	Hence, given $L>0$ we have that
	\[
		\inf\left\{\lambda_{1,p}(\Gamma,\V_D)\colon \Gamma 
		\text{ is a star graph}, \ell(\Gamma)=L,
		\emptyset\neq\V_D\subset\V(\Gamma)
		\right\}
	\]
	is equal to $\left(\dfrac{\pi_{p}}{2L}\right)^p
		\dfrac{p}{p^\prime}.$
\end{ex}

\medskip

Finally, we give an upper bound for the first eigenvalue of the $p-$Laplacian. 

\begin{teo}
	\label{teo:cotasup} Let $\Gamma$ be a connected compact metric graph,
	$\V_D$ be a non-empty subset of $\V(\Gamma)$ and
	$p\in(1,+\infty).$
	Then
	\[
		\lambda_{1,p}(\Gamma,\V_D)\le 
		\left(\dfrac{\mathrm{card}(E(\Gamma))\pi_p}{\ell(\Gamma)}\right)^p
		\dfrac{p}{p^\prime},
	\]
	where $\mathrm{card}(E(\Gamma))$ is the number of elements in $\E(\Gamma).$
\end{teo}

\begin{proof} 
	Let $e_0\in \E(\Gamma)$ such that
	$\ell_{e_0}=\max\{\ell_{e}\colon e\in\E(\Gamma)\}.$ Then
	\begin{equation}\label{eq:cs1}
		\ell_{e_0}\ge\dfrac{\mathrm{card}(E(\Gamma))}{\ell(\Gamma)}.
	\end{equation}	
	
	On the other hand, taking
	\[
		u(x)=
		\begin{cases}
			\dfrac{\ell_{e_0}}{\pi_p}\sin_p
			\left(\dfrac{\pi_p}{\ell_{e_0}}
			x\right) 
			&\text{ if } x\in I_{e_0}\\
			0 &\text{otherwise},
		\end{cases}
	\]
	and using \eqref{eq:cs1}, we have that
	\[
		\lambda_{1,p}(\Gamma,V_D)\le
		\dfrac{\displaystyle\int_{e_0}|u^\prime(x)|\, dx}
		{\displaystyle\int_{e_0}|u^\prime(x)|\, dx}=
		\left(\dfrac{\pi_{p}}{\ell_{e_0}}\right)^p
		\dfrac{p}{p^\prime}\le 
		\left(\dfrac{\mathrm{card}(E(\Gamma))\pi_p}{\ell(\Gamma)}\right)^p
		\dfrac{p}{p^\prime} .
	\]
	This completes the proof.\qed
\end{proof}

 The upper bound is also optimal.
\begin{ex}
	Let $\Gamma$ as in Eample \ref{ex:cotainf}
	 and  $\V_D=\{\v_1,\v_2\}.$
	
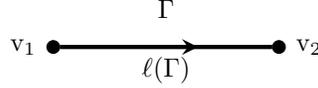
	\begin{center}
	\begin{tikzpicture}
			 \node (a0) at (1.5,.5) {$\Gamma$};
  		    \node  [fill, circle,draw, scale=.5]  (a) at (0,0) {};
  			\node  (a1) at (1.5,-0.3)  { $\ell(\Gamma)$ };
  			\node  (c) at (1.5,0)  {};
  			\node  [fill, circle,draw, scale=.5]  (b) at (3,0)  {};
  			\node  (a2) at (-0.4,0)  { $\v_1$ };
  			\node  (b1) at (3.4,0)  { $\v_2$ };
  			\draw[directed,ultra thick] (a) -- (b);	
		\end{tikzpicture}
	\end{center}
	Then
	\[
		\mathrm{card}(E(\Gamma))=1 \quad\mbox{ and }\quad
		\lambda_{1,p}(\Gamma,V_D)=
		\left(\dfrac{\pi_{p}}{\ell(\Gamma)}\right)^p
		\dfrac{p}{p^\prime}.
	\]
\end{ex}

\begin{ex}\label{ex:cotasup*}
	Let $\Gamma$ be a star graph with $n+1$ vertices
	 and $n$ edges, that is $\V(\Gamma)=\{\v_0,\v_1,\dots,\v_n\}$ and 
	 $\E(\Gamma)=\{[\v_1,\v_0], [\v_2,\v_0], \dots, [\v_n,\v_0]\}.$ Let
	 $\V_D=\V(\Gamma)$ and $\ell([\v_i,\v_0])=\ell$ for all 
	 $i\in\{1,\dots, n\}.$
	 Then $\ell(\Gamma)=n\ell=\mathrm{card}(\E(\Gamma))\ell.$
	 
	 
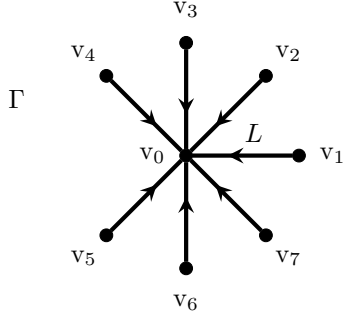
\medskip
	 
	\begin{center}
	\begin{tikzpicture}[scale=1.5]
			\node 	(a0) at (-1.5,.5) {$\Gamma$};
  		    \node  [fill, circle,draw, scale=.5]  (a) at (0,0) {};
  			\node  (a1) at (.6,0.2)  { $L$ };
  			\node  (c) at (1.5,0)  {};
  			\node  [ fill, circle,draw, scale=.5]  (b) at (1,0)  {};
  			\node  [ fill, circle,draw, scale=.5]  (c) at 
  			(0.70710678118,0.70710678118)  {};
  			\node  [ fill, circle,draw, scale=.5]  (d) at 
  			(0,1)  {};
  			\node  [ fill, circle,draw, scale=.5]  (e) at 
  			(-0.70710678118,0.70710678118)  {};
  			\node  [ fill, circle,draw, scale=.5]  (f) at 
  			(0.70710678118,-0.70710678118)  {};
  			\node  [ fill, circle,draw, scale=.5]  (g) at 
  			(-0.70710678118,-0.70710678118)  {};
  			\node  [ fill, circle,draw, scale=.5]  (h) at 
  			(0,-1)  {};
  			\node  (b0) at (-0.3,0)  { $\v_0$ };
  			\node  (b1) at (1.3,0)  { $\v_1$ };
  			\node  (b2) at (0,1.3)  { $\v_3$ };
  			\node  (b3) at (0,-1.3)  { $\v_6$ };
  			\node  (c1) at (0.90710678118,0.90710678118)  { $\v_2$ };
  			\node  (c2) at (-0.90710678118,0.90710678118)  { $\v_4$ };
  			\node  (c3) at (-0.90710678118,-0.90710678118)  { $\v_5$ };
  			\node  (c3) at (0.90710678118,-0.90710678118)  { $\v_7$ };
  			\draw[directed,ultra thick] (b) -- (a);	
  			\draw[directed,ultra thick] (c) -- (a);	
  			\draw[directed,ultra thick] (d) -- (a);	
  			\draw[directed,ultra thick] (e) -- (a);	
  			\draw[directed,ultra thick] (f) -- (a);	
  			\draw[directed,ultra thick] (g) -- (a);	
  			\draw[directed,ultra thick] (h) -- (a);	
		\end{tikzpicture}
	\end{center}
	Then
	\[
		\lambda_{1,p}(\Gamma,V_D)=
		\left(\dfrac{\pi_{p}}{\ell}\right)^p
		\dfrac{p}{p^\prime}=\left(
		\dfrac{\mathrm{card}(\E(\Gamma))\pi_{p}}{\ell(\Gamma)}\right)^p
		\dfrac{p}{p^\prime}.
	\]
	Hence, given $L>0$ and $n\in\mathbb{N}$ we have that
	\[
		\max\left\{\lambda_{1,p}(\Gamma,\V_D)\colon \Gamma 
		\text{ is a star graph}, \ell(\Gamma)=L, 
		\mathrm{card}(\E(\Gamma))=n,
		\emptyset\neq\V_D
		\right\}
	\]
	is equal to $\left(\dfrac{n\pi_{p}}{L}\right)^p
		\dfrac{p}{p^\prime}.$
\end{ex}

\section{The shape derivative of $\lambda_{1,p}(\Gamma,\V_D)$.} \label{sect-shape}

The aim of this section is to study the perturbation properties
of $\lambda_{1,p}(\Gamma,\V_D)$ with respect to the edges. 

More precisely, let 
$e_0\in\E(\Gamma)$ such that $e_0=[\u,\v],$
we consider the following family of graphs 
$\{\Gamma_\delta\}_{\delta\in\mathbb{R}}$ where for any $\delta$
\[
	\V(\Gamma_\delta)=\V(\Gamma),	
	\quad \E(\Gamma_{\delta})= \E(\Gamma_{\delta})
\]
and the length assigned to $e\in\E(\Gamma_{\delta})$ is
\[
	\ell^\delta_e= \begin{cases}
		\ell_{e_0}+\delta, &\text{ if } e=e_0,\\
		\ell_e, &\text{ otherwise. }\\
\end{cases}
\] 
The problem of perturbation of eigenvalues consists in analyzing the 
dependence of $\lambda(\delta)\coloneqq\lambda_{1,p}(\Gamma_{\delta},\V_D)$
with respect to $\delta.$ Note that $\lambda(0)=\lambda_{1,p}(\Gamma,\V_D).$

\begin{lem}\label{lem:continuidad}
	Let $\Gamma$ be a connected compact metric graph, $\V_D$ 
	be a non-empty subset of $\V(\Gamma)$ and $p\in(1,+\infty).$ 
	Then function $\lambda(\delta)$ is
	continuous at $\delta=0.$ 
\end{lem} 
	
\begin{proof}
	Let $u$ be an eigenfunction associated to $\lambda(0)$  with
	$\|u\|_{L^p(\Gamma)}=1.$ Then
	\[
		w_\delta(x)=
		\begin{cases} 
			u\left(\dfrac{\ell_{e_0}}{\ell_{e_0}+\delta}x\right)
			&\text{ if } x\in I_{e_0},\\
			u(x) &\text{ otherwise},
		\end{cases}
	\]
	belongs to $\X(\Gamma_\delta,\V_D)$ for all $\delta.$ Therefore
	for any $\delta$
	\begin{align*}
		\lambda(\delta)&\le\dfrac{\displaystyle\int_{\Gamma_\delta} 
		|w_\delta^\prime(x)|^p\, dx}{\displaystyle\int_{\Gamma_\delta} 
		|w_\delta(x)|^p\, dx}\\
		&=\dfrac{\displaystyle\sum_{e\in\E(\Gamma)\setminus\{e_0\}}
		\int_{0}^{\ell_{e}} |u^\prime(x)|^p\, dx + 
		\int_{0}^{\ell_{e_0}+\delta} 
		\left|u^\prime\left(\dfrac{\ell_{e_0}}{\ell_{e_0}+\delta}x\right)
		\right|^p\left(\dfrac{\ell_{e_0}}{\ell_{e_0}+\delta}\right)^{p}
		\, dx}
		{\displaystyle\sum_{e\in\E(\Gamma)\setminus\{e_0\}}
		\int_{0}^{\ell_{e}} |u(x)|^p\, dx + 
		\int_{0}^{\ell_{e_0}+\delta} 
		\left|u\left(\dfrac{\ell_{e_0}}{\ell_{e_0}+\delta}x\right)
		\right|^p\, dx}\\
		&=\dfrac{\displaystyle\sum_{e\in\E(\Gamma)\setminus\{e_0\}}
		\int_{0}^{\ell_{e}} |u^\prime(x)|^p\, dx + 
		\int_{0}^{\ell_{e_0}} 
		\left|u^\prime(x)\, \right|^p
		dx
		\left(\dfrac{\ell_{e_0}}{\ell_{e_0}+\delta}\right)^{p-1}
		}
		{\displaystyle\sum_{e\in\E(\Gamma)\setminus\{e_0\}}
		\int_{0}^{\ell_{e}} |u(x)|^p\, dx + 
		\int_{0}^{\ell_{e_0}} 
		\left|u(x)
		\right|^p\, dx\dfrac{\ell_{e_0}+\delta}{\ell_{e_0}}}.
	\end{align*} 
	Since $u$ is an eigenfunction associated to $\lambda(0)$ and 
	$\|u\|_{L^p(\Gamma)} =1,$ we have that
	\begin{equation}\label{eq:cont1}
		\lambda(\delta)\le
		\dfrac{\lambda(0)+\displaystyle
		\left[\left(\dfrac{\ell_{e_0}}{\ell_{e_0}+\delta}\right)^{p-1}
		-1\right]\int_{0}^{\ell_{e_0}} 
		\left|u^\prime(x)\, \right|^p
		dx}{1+\displaystyle\dfrac{\delta}{\ell_{e_0}}
		\int_{0}^{\ell_{e_0}} \left|u(x)\right|^p\, dx}
		\qquad\forall\delta.
	\end{equation}
	Therefore
	\begin{equation}\label{eq:cont2}
		\limsup_{\delta\to0}\lambda(\delta)\le\lambda(0).
	\end{equation}
	
	Then to show that $\lambda(\delta)$ is continuous at $\lambda=0,$ 
	it remains to prove that
	\begin{equation}\label{eq:cont3}
		\liminf_{\delta\to0}\lambda(\delta)\ge\lambda(0).
	\end{equation}
	
	Let $u_{\delta}$ be an eigenfunction 
	associated to $\lambda(\delta)$ normalized by 
	$\|u_\delta\|_{L^p(\Gamma_\delta)}=1.$ 
	Then, for any $\delta$
	\[
		v_{\delta}(x)=
		\begin{cases} 
			u_{\delta}\left(\dfrac{\ell_{e_0}+\delta}{\ell_{e_0}}x\right)
			&\text{ if } x\in I_{e_0},\\
			u_{\delta}(x) &\text{ otherwise, }
		\end{cases}
	\]
	belongs to $\X(\Gamma,\V_D).$ Moreover 
	\begin{equation}\label{eq:cont4}
		\begin{array}{l}
		\displaystyle
			\|v_\delta\|_{L^p(\Gamma)}^p=
			\int_{\Gamma}|v_\delta(x)|^p\, dx\\
			\displaystyle
			=\sum_{e\in\E(\Gamma)\setminus\{e_0\}}
			\int_{0}^{\ell_{e}}|u_\delta(x)|^p\, dx
			+\int_{0}^{\ell_{e_0}}\left
			|u_\delta\left(\dfrac{\ell_{e_0}+\delta}{\ell_{e_0}}
			x\right)\right|^p\, dx\\
			\displaystyle =\sum_{e\in\E(\Gamma)\setminus\{e_0\}}
			\int_{0}^{\ell_{e}}|u_\delta(x)|^p\, dx
			+\left(\dfrac{\ell_{e_0}}{\ell_{e_0}+\delta}\right)
			\int_{0}^{\ell_{e_0}+\delta}\left
			|u_\delta(x)\right|^p
			dx\\
			\displaystyle =  1-\dfrac{\delta}{\ell_{e_0}+\delta}
			\int_{0}^{\ell_{e_0}+\delta}\left
			|u_\delta(x)\right|^p
			\, dx
		\end{array}
	\end{equation}
	for all $\delta,$ and
	$$
	\begin{array}{l}
	\displaystyle
		\|v_\delta^\prime\|_{L^p(\Gamma)}^p=
			\int_{\Gamma}|v_\delta^\prime(x)|^p\, dx\\
			\displaystyle
			=\sum_{e\in\E(\Gamma)\setminus\{e_0\}}
			\int_{0}^{\ell_{e}}|u_\delta^\prime(x)|^p\, dx
			+\int_{0}^{\ell_{e_0}}\left
			|u_\delta^\prime\left(\dfrac{\ell_{e_0}+\delta}{\ell_{e_0}}
			x\right)\right|^p
			\left(\dfrac{\ell_{e_0}+\delta}{\ell_{e_0}}\right)^p dx\\
			\displaystyle =\sum_{e\in\E(\Gamma)\setminus\{e_0\}}
			\int_{0}^{\ell_{e}}|u_\delta^\prime(x)|^p\, dx
			+\left(1+\dfrac{\delta}{\ell_{e_0}}\right)^{p-1}
			\int_{0}^{\ell_{e_0}+\delta}\left
			|u_\delta^\prime(x)\right|^p
			\, dx.
	\end{array}
	$$
	Hence
	\begin{equation}\label{eq:cont5}
		\|v_\delta^\prime\|_{L^p(\Gamma)}^p=\lambda(\delta)+
		\left[\left(1+\dfrac{\delta}{\ell_{e_0}}\right)^{p-1}
		-1\right]\int_{0}^{\ell_{e_0}+\delta_j}\left
			|u_\delta^\prime(x)\right|^p
			\, dx.\\
	\end{equation}
Then
	\begin{equation}\label{eq:cont6}
		\lambda(0)\le\dfrac{\displaystyle\lambda(\delta)+
		\left[\left(1+\dfrac{\delta}{\ell_{e_0}}\right)^{p-1}
		-1\right]\int_{0}^{\ell_{e_0}+\delta_j}\left
			|u_\delta^\prime(x)\right|^p
			\, dx}{\displaystyle1-\dfrac{\delta}{\ell_{e_0}+\delta}
			\int_{0}^{\ell_{e_0}+\delta}\left
			|u_\delta(x)\right|^p
			\, dx}
	\end{equation}
	for all $\delta.$
	
	Let $\{\delta_{j}\}_{j\in\mathbb{N}}$ such that $\delta_j\to0$ as
	$j\to \infty$ and
	\begin{equation}\label{eq:cont30}
		\lim_{j\to +\infty}\lambda(\delta_j)=
		\liminf_{\delta\to0}\lambda(\delta).
	\end{equation}
	Then, by \eqref{eq:cont1}, \eqref{eq:cont30},
	\eqref{eq:cont4}, and \eqref{eq:cont5},
	$\{v_{\delta_j}\}_{j\in\mathbb{N}}$ 
	is bounded in $W^{1,p}(\Gamma).$ Hence there exist
	a subsequence (still denote $\{v_{\delta_j}\}_{j\in\mathbb{N}}$) 
	and $u_0\in\X(\Gamma,\V_D)$ such that
	\begin{align*}
		v_{\delta_j}
		\rightharpoonup u_0 &\quad\text{ weakly in } W^{1,p}(\Gamma),\\
		v_{\delta_j}\to u_0 &\quad\text{ strongly in } L^{p}(\Gamma).  
	\end{align*}
	Then, by \eqref{eq:cont4}, we have $\|u_0\|_{L^p(\Gamma)}=1.$ 
	In addition, by \eqref{eq:cont5} and \eqref{eq:cont30}, we get 
	\begin{align*}
		\lambda(0)&\le \int_\Gamma |u_0^{\prime}(x)|^p
		\, dx\\
		&\le
		\liminf_{j\to+\infty}
		\int_\Gamma |v_{\delta_j}^{\prime}(x)|^p
		\, dx\\
		&\le
		\liminf_{j\to+\infty}
		\lambda(\delta_j)+
		\left[\left(1+\dfrac{\delta_j}{\ell_{e_0}}\right)^{p-1}
		-1\right]\int_{0}^{\ell_{e_0}+\delta_j}\left
			|u_{\delta_j}^\prime(x)\right|^p
			\, dx\\
		&=\liminf_{\delta\to 0}\lambda(\delta).
	\end{align*}
	Therefore \eqref{eq:cont3} holds. 
	
	Thus, by \eqref{eq:cont2} and \eqref{eq:cont3}, 
	the function $\lambda(\delta)$ is
	continuous at $\delta=0.$\qed	
\end{proof} 

\begin{co}\label{co:continuidad}
	Let $\Gamma$ be a connected compact metric graph, $\V_D$ 
	be a non-empty subset of $\V(\Gamma),$ $p\in(1,+\infty)$ and
	$u_\delta$ be an eigenfunction 
	associated to $\lambda(\delta)$ normalized by 
	$\|u_\delta\|_{L^p(\Gamma_\delta)}=1.$ Then there exists a
	subsequence $\delta_j\to 0$ and an eigenfunction $u_0$ associated
	to $\lambda(0)$ such that
	\[
		v_{\delta_j}\to u_0 \mbox{ strongly in } \X(\Gamma,\V_D)
	\]
	as $j\to +\infty$ where 
	\[
		v_{\delta_j}(x)=
		\begin{cases} 
			u_{\delta_j}\left(\dfrac{\ell_{e_0}+\delta}{\ell_{e_0}}x\right)
			&\text{ if } x\in I_{e_0}\\
			u_{\delta_j}(x) &\text{ otherwise. }
		\end{cases}
	\]
	Moreover $\|u_0\|_{L^p(\Gamma)}=1$ and
	\begin{align*}
	 	 \lim_{j\to\infty}
		 \int_{0}^{\ell_{e_0}+\delta_j}\left
			|u_{\delta_j}(x)\right|^p
			\, dx&=\int_{0}^{\ell_{e_0}} |u_0 (x) |^p\, dx,\\
		\lim_{j\to\infty}
		\int_{0}^{\ell_{e_0}+\delta_j}\left
			|u_j^\prime(x)\right|^p
			\, dx&=\int_{0}^{\ell_{e_0}} |u_0^\prime (x)|^p\, dx.
	\end{align*}
\end{co}
\begin{proof}
	Let $\{\delta_j\}_{j\in\mathbb{N}}$ such that $\delta_j\to 0$ as 
	$j\to \infty.$ By \eqref{eq:cont4} and \eqref{eq:cont5}, 
	we have that 
	\begin{align*}
			\|v_{\delta_j}\|_{L^p(\Gamma)}^p&=
			1-\dfrac{\delta_j}{\ell_{e_0}+\delta_j}
			\int_{0}^{\ell_{e_0}+\delta_j}\left
			|u_j(x)\right|^p
			\, dx\\
		\|v_{\delta_j}^\prime\|_{L^p(\Gamma)}^p&=\lambda(\delta_j)+
		\left[\left(1+\dfrac{\delta_j}{\ell_{e_0}}\right)^{p-1}
		-1\right]\int_{0}^{\ell_{e_0}+\delta_j}\left
			|u_{\delta_j}^\prime(x)\right|^p
			\, dx.
	\end{align*}
	By Lemma \ref{lem:continuidad}, we have that $\lambda(\delta_j)\to \lambda(0).$
	Then $\{v_{\delta_j}\}_{j\in\mathbb{N}}$ is bounded in $\X(\Gamma,V_D),$
	\begin{align}
		\label{eq:cocont1}&\|v_{\delta_j}\|_{L^p(\Gamma)}^p\to 1,\\
		\label{eq:cocont2}&\|v_{\delta_j}\|_{L^p(\Gamma)}^p\to \lambda(0),
	\end{align}
	as $j\to\infty.$
	Therefore there exists
	a subsequence (still denoted $\{v_{\delta_j}\}_{j\in\mathbb{N}}$) 
	and $u_0\in\X(\Gamma,\V_D)$ such that
	\begin{align*}
		v_{\delta_j}
		\rightharpoonup u_0 &\quad\text{ weakly in } W^{1,p}(\Gamma),\\
		v_{\delta_j}\to u_0 &\quad\text{ strongly in } L^{p}(\Gamma).  
	\end{align*}
	Then, by \eqref{eq:cont1}, we have $\|u_0\|_{L^p(\Gamma)}=1.$ 
	In addition, by \eqref{eq:cont2}, 
	we get 
	\[
		\lambda(0)=\int_\Gamma |u_0^{\prime}(x)|^p
		\, dx
		\le
		\liminf_{j\to+\infty}
		\int_\Gamma |v_j^{\prime}(x)|^p
		\, dx=\lambda(0).
	\]
	Therefore $u_0$ is an eigenfunction associated to $\lambda(0)$ and
	$$\|v_{\delta_j}\|_{W^{1,p}(\Gamma)}\to 
	\|u_{0}\|_{W^{1,p}(\Gamma)}$$ as $j\to\infty.$ Since 
	$v_{\delta_j}\rightharpoonup u_0$  weakly in  $W^{1,p}(\Gamma),$
	we have that $v_{\delta_j}\to u_0$  strongly in  
	$W^{1,p}(\Gamma).$ Then $v_{\delta_j}\to u_0$ strongly in  
	$W^{1,p}(I_{e_0})$ and hence
	\begin{align*}
		\int_{0}^{\ell_{e_0}} |u_0|^p\, dx&=\lim_{j\to \infty}
		\int_{0}^{\ell_{e_0}} |v_{\delta_j}|^p\, dx
		=\lim_{j\to\infty}
		\left(\dfrac{\ell_{e_0}}{\ell_{e_0}+\delta_j}\right)
		\int_{0}^{\ell_{e_0}+\delta_j}\left
			|u_{\delta_j}(x)\right|^p
			\, dx,\\
		\int_{0}^{\ell_{e_0}} |u_0^\prime|^p\, dx&=\lim_{j\to \infty}
		\int_{0}^{\ell_{e_0}} |v_{\delta_j}^\prime|^p dx
		=\lim_{j\to\infty}
		\left(1+\dfrac{\delta_j}{\ell_{e_0}}\right)^{p-1}
		\int_{0}^{\ell_{e_0}+\delta_j}\left
			|u_{\delta_j}^\prime(x)\right|^p
			 dx,
\end{align*}
	that is
	\begin{align*}
	 	\int_{0}^{\ell_{e_0}} |u_0(x)|^p\, dx&= \lim_{j\to\infty}
		 \int_{0}^{\ell_{e_0}+\delta_j}\left
			|u_{\delta_j}(x)\right|^p
			 dx,\\
		\int_{0}^{\ell_{e_0}} |u_0^\prime (x)|^p\, dx&=\lim_{j\to\infty}
		\int_{0}^{\ell_{e_0}+\delta_j}\left
			|u_j^\prime(x)\right|^p
			 dx,
	\end{align*}
	which completes the proof.\qed
\end{proof}

Before proving that the function $\lambda$ is differentiable at $\delta=0$ 
when the first eigenvalue is simple, we will show that, 
in the general case, $\lambda$ is differentiable from the left and from the right  at 
$\delta=0.$

\begin{lem}\label{lem:derivabilidad}
	Let $\Gamma$ be a connected compact metric graph, $\V_D$ 
	be a non-empty subset of $\V(\Gamma)$ and $p\in(1,+\infty).$ 
	Then the function $\lambda(\delta)$ is left and right 
	differentiable at $\delta=0$ and 
	\begin{align*}
		&\lim_{\delta\to0^+}\dfrac{\lambda(\delta)-\lambda(0)}{\delta}=
	 	\min_{u\in \mathcal{E}}
	 	\left\{-\dfrac{(p-1)}{\ell_{e_0}}\int_{0}^{\ell_{e_0}} 
		\left|u_0^\prime\, \right|^p 
		-\dfrac{\lambda(0)}{\ell_{e_0}}
		\int_{0}^{\ell_{e_0}} \left|u_0\right|^p
		\right\},\\
		&\lim_{\delta\to0^-}\dfrac{\lambda(\delta)-\lambda(0)}{\delta}=
	 	\max_{u\in \mathcal{E}}
	 	\left\{-\dfrac{(p-1)}{\ell_{e_0}}\int_{0}^{\ell_{e_0}} 
		\left|v_0^\prime \, \right|^p 
		-\dfrac{\lambda(0)}{\ell_{e_0}}
		\int_{0}^{\ell_{e_0}} \left|v_0\right|^p\right\},
	\end{align*}
where $\mathcal{E}$ is the set of eigenfunctions $u$
	associated to $\lambda(0)$ normalized 
	with $\|u\|_{L^p(\Gamma)}=1.$
\end{lem} 

\begin{proof}
	We split the proof in several steps.

	{\it Step 1.} We start by showing that
	\[
		\limsup_{\delta\to0^+}\dfrac{\lambda(\delta)-\lambda(0)}{\delta}
		\le -\dfrac{(p-1)}{\ell_{e_0}}\int_{0}^{\ell_{e_0}} 
		\left|u^\prime(x)\, \right|^p dx
		-\dfrac{\lambda(0)}{\ell_{e_0}}
		\int_{0}^{\ell_{e_0}} \left|u(x)\right|^p\, dx
	\]
	for any eigenfunction $u$ associated to $\lambda(0)$ normalized
	by $\|u\|_{L^p(\Gamma)}=1.$

	Let $u$ be an  eigenfunction associated to $\lambda(0)$ normalized
	by $\|u\|_{L^p(\Gamma)}=1.$ By \eqref{eq:cont1}, we have
	\[
		\lambda(\delta)-\lambda(0)\le
		\dfrac{\displaystyle
		\left[\left(\dfrac{\ell_{e_0}}{\ell_{e_0}+\delta}\right)^{p-1}
		-1\right]\int_{0}^{\ell_{e_0}} 
		\left|u^\prime(x)\, \right|^p
		dx-\lambda(0)\dfrac{\delta}{\ell_{e_0}}
		\int_{0}^{\ell_{e_0}} \left|u(x)\right|^p\, dx}
		{1+\displaystyle\dfrac{\delta}{\ell_{e_0}}
		\int_{0}^{\ell_{e_0}} \left|u(x)\right|^p\, dx}
	\]
	for all $\delta.$ Then
	\[
		\dfrac{\lambda(\delta)-\lambda(0)}{\delta}\le
		\dfrac{\displaystyle
		\dfrac{\left(\dfrac{\ell_{e_0}}{\ell_{e_0}+\delta}\right)^{p-1}
		-1}{\delta}\int_{0}^{\ell_{e_0}} 
		\left|u^\prime(x)\, \right|^p
		dx-\dfrac{\lambda(0)}{\ell_{e_0}}
		\int_{0}^{\ell_{e_0}} \left|u(x)\right|^p\, dx}
		{1+\displaystyle\dfrac{\delta}{\ell_{e_0}}
		\int_{0}^{\ell_{e_0}} \left|u(x)\right|^p\, dx}
	\]
	for all $\delta>0.$ Therefore
	\[
		\limsup_{\delta\to0^+}\dfrac{\lambda(\delta)-\lambda(0)}{\delta}
		\le -\dfrac{(p-1)}{\ell_{e_0}}\int_{0}^{\ell_{e_0}} 
		\left|u^\prime(x)\, \right|^p dx
		-\dfrac{\lambda(0)}{\ell_{e_0}}
		\int_{0}^{\ell_{e_0}} \left|u(x)\right|^p\, dx.
	\]
	
	\medskip
	
	{\it Step 2.} With a similar procedure, we obtain
	\[
		\liminf_{\delta\to0^-}\dfrac{\lambda(\delta)-\lambda(0)}{\delta}
		\ge -\dfrac{(p-1)}{\ell_{e_0}}\int_{0}^{\ell_{e_0}} 
		\left|u^\prime(x)\, \right|^p dx
		-\dfrac{\lambda(0)}{\ell_{e_0}}
		\int_{0}^{\ell_{e_0}} \left|u(x)\right|^p\, dx.
	\]
	for any eigenfunction $u$ associated to $\lambda(0)$ normalized
	by $\|u\|_{L^p(\Gamma)}=1.$
	
	\medskip
	
	{\it Step 3.} Now we show that there exists an eigenfunction
	$u_0$ associated to $\lambda(0)$ normalized by $\|u_0\|_{L^p(\Gamma)}=1$
	such that 
	\[
		\liminf_{\delta\to0^+}\dfrac{\lambda(\delta)-\lambda(0)}{\delta}
		\ge -\dfrac{(p-1)}{\ell_{e_0}}\int_{0}^{\ell_{e_0}} 
		\left|u_0^\prime(x)\, \right|^p dx
		-\dfrac{\lambda(0)}{\ell_{e_0}}
		\int_{0}^{\ell_{e_0}} \left|u_0(x)\right|^p\, dx.
	\]
	
	Let $u_{\delta}$ be an eigenfunction 
	associated to $\lambda(\delta)$ normalized by 
	$\|u_\delta\|_{L^p(\Gamma_\delta)}=1.$ By \eqref{eq:cont6}, we have
	\[
		\lambda(\delta)-\lambda(0)\ge
		-\dfrac{A(\delta)}{B(\delta)}
		\qquad \forall\delta
	\]
	where
	\begin{align*}
		A(\delta) &=\dfrac{\lambda(\delta)\delta}
		{(\ell_{e_0}+\delta)}
			\int_{0}^{\ell_{e_0}+\delta}\left
			|u_\delta(x)\right|^p
			 dx-
		\left[\left(1+\nicefrac{\delta}{\ell_{e_0}}\right)^{p-1}
		-1\right]\int_{0}^{\ell_{e_0}+\delta_j}\left
			|u_\delta^\prime(x)\right|^p
			 dx\\
			B(\delta)&=1-\dfrac{\delta}{(\ell_{e_0}+\delta)}
			\int_{0}^{\ell_{e_0}+\delta}\left
			|u_\delta(x)\right|^p
			 dx.
	\end{align*}
	Then
	\begin{equation}\label{eq:deriv0}
			\dfrac{\lambda(\delta)-\lambda(0)}{\delta}\ge
			\dfrac{\dfrac{A(\delta)}{\delta}}{B(\delta)}
	\end{equation}
	for all $\delta>0.$
	Let $\{\delta_{j}\}_{j\in\mathbb{N}}$ such that $\delta_j\to0^+$ as
	$j\to \infty$ and
	\begin{equation}\label{eq:deriv1}
		\lim_{j\to +\infty}\dfrac{\lambda(\delta_j)-\lambda(0)}{\delta_j}
		=\liminf_{\delta\to0^+}
		\dfrac{\lambda(\delta)-\lambda(0)}{\delta}.
	\end{equation}
	Then, by Corollary \ref{co:continuidad}, there exist a
	subsequence (still denoted $\delta_j$) and an eigenfunction 
	$u_0$ associated to $\lambda(0)$ such that
	\begin{align*}
		\|u_0\|_{L^p(\Gamma)}&=1,\\
	 	 \lim_{j\to\infty}
		 \int_{0}^{\ell_{e_0}+\delta_j}\left
			|u_{\delta_j}(x)\right|^p
			\, dx&=\int_{0}^{\ell_{e_0}} |u_0 (x)|^p\, dx,\\
		\lim_{j\to\infty}
		\int_{0}^{\ell_{e_0}+\delta_j}\left
			|u_j^\prime(x)\right|^p
			\, dx&=\int_{0}^{\ell_{e_0}} |u_0^\prime (x)|^p\, dx.
	\end{align*}
	Therefore 
	\begin{align*}
		\lim_{j\to+\infty}\dfrac{A(\delta_j)}{\delta_j}&=
		-\dfrac{(p-1)}{\ell_{e_0}}\int_{0}^{\ell_{e_0}} 
		\left|u_0^\prime(x)\, \right|^p dx
		-\dfrac{\lambda(0)}{\ell_{e_0}}
		\int_{0}^{\ell_{e_0}} \left|u_0(x)\right|^p\, dx,\\
		\lim_{j\to+\infty}B(\delta_j)&=1.
	\end{align*}
	In addition, by \eqref{eq:deriv0} and \eqref{eq:deriv1}, we get
	\[
		\liminf_{\delta\to0^+}\dfrac{\lambda(\delta)-\lambda(0)}{\delta}
		\ge -\dfrac{(p-1)}{\ell_{e_0}}\int_{0}^{\ell_{e_0}} 
		\left|u_0^\prime(x)\, \right|^p dx
		-\dfrac{\lambda(0)}{\ell_{e_0}}
		\int_{0}^{\ell_{e_0}} \left|u_0(x)\right|^p\, dx.
	\]
	Hence, by {\it step 1}, we have that
	\[
		\lim_{\delta\to0^+}\dfrac{\lambda(\delta)-\lambda(0)}{\delta}
		= -\dfrac{(p-1)}{\ell_{e_0}}\int_{0}^{\ell_{e_0}} 
		\left|u_0^\prime(x)\, \right|^p dx
		-\dfrac{\lambda(0)}{\ell_{e_0}}
		\int_{0}^{\ell_{e_0}} \left|u_0(x)\right|^p\, dx.
	\]

	{\it Step 4.} In the same way, we can show that
	there exists an eigenfunction
	$v_0$ associated to $\lambda(0)$ such that 
	\[
		\lim_{\delta\to0^-}\dfrac{\lambda(\delta)-\lambda(0)}{\delta}
		= -\dfrac{(p-1)}{\ell_{e_0}}\int_{0}^{\ell_{e_0}} 
		\left|v_0^\prime(x)\, \right|^p dx
		-\dfrac{\lambda(0)}{\ell_{e_0}}
		\int_{0}^{\ell_{e_0}} \left|v_0(x)\right|^p\, dx.
	\]\qed
\end{proof}

Thus, if the first eigenvalue is simple then the function $\lambda(\delta)$ 
is differentiable at $\delta=0.$

\begin{teo}\label{teo:derivabilidad}
	Let $\Gamma$ be a connected compact metric graph, $\V_D$ 
	be a non-empty subset of $\V(\Gamma)$ and $p\in(1,+\infty).$ 
	If the first eignevalue $\lambda_{1,p}(\Gamma,\V_D)$ is simple, then 
	the function $\lambda(\delta)$ is
	differentiable at $\delta=0$ and 
	\[
	 	\lambda^\prime(0)=
	 	-\dfrac{(p-1)}{\ell_{e_0}}\int_{0}^{\ell_{e_0}} 
		\left|u_0^\prime(x)\, \right|^p dx
		-\dfrac{\lambda(0)}{\ell_{e_0}}
		\int_{0}^{\ell_{e_0}} \left|u_0(x)\right|^p\, dx
	\]
	where $u_0$ is an eigenfunction  associated to $\lambda(0)$ normalized
	by $\|u\|_{L^p(\Gamma)}=1.$
\end{teo} 

\begin{re}  
	Note that the result of Theorem \ref{teo:derivabilidad} 
	does not hold if we remove the assumption that the 
	first eigenvalue is simple.
	For example, let $\Gamma$  defined as in Remark \ref{re:nosimple}
	and $e_0=[\v_2,\v_3]$ we have that
	\begin{align*}
		\lim_{t\to0^+}\dfrac{\lambda(\delta)-\lambda(0)}{\delta}&=
	 	\min_{u\in \mathcal{E}}
	 	\left\{-\dfrac{(p-1)}{L}\int_{0}^{\ell_{e_0}} 
		\left|u_0^\prime \, \right|^p 
		-\dfrac{\lambda(0)}{L}
		\int_{0}^{\ell_{e_0}} \left|u_0\right|^p 
		\right\}\\
		&=\min_{u\in \mathcal{E}}
	 	\left\{-\dfrac{p\lambda(0)}{L}\int_{0}^{\ell_{e_0}} 
		\left|u_0^\prime \, \right|^p 
		\right\}
		=-\dfrac{p\lambda(0)}{L},
	\end{align*}
	\begin{align*}
		\lim_{t\to0^-}\dfrac{\lambda(\delta)-\lambda(0)}{\delta}&=
	 	\max_{u\in \mathcal{E}}
	 	\left\{-\dfrac{(p-1)}{\ell_{e_0}}\int_{0}^{\ell_{e_0}} 
		\left|v_0^\prime \, \right|^p 
		-\dfrac{\lambda(0)}{\ell_{e_0}}
		\int_{0}^{\ell_{e_0}} \left|v_0 \right|^p \right\}\\
		&=\max_{u\in \mathcal{E}}
	 	\left\{-\dfrac{p\lambda(0)}{L}\int_{0}^{\ell_{e_0}} 
		\left|u_0^\prime \, \right|^p 
		\right\} =0.
	\end{align*}
	Hence $\lambda$ is not differentiable (but Lipschitz) at $\delta=0.$
\end{re} 

\section{The limit as $p \to \infty$.} \label{sect-infty}

In this section we deal with the limit as $p\to \infty$ of the eigenvalue 
problem \eqref{eq:autovalor}.

\begin{teo} \label{teo:infty} 
	Let $\Gamma$ be a connected compact metric graph,
	$\V_D$ be a non-empty subset of $\V(\Gamma),$ and
	$u_p$ be a minimizer for \eqref{eq:autovalor} normalized by 
	$\|u_p\|_{L^p(\Gamma)}=1.$ 
	Then, there exists a sequence $p_j \to \infty$ such that 
	\[
		u_{p_j} \to u_\infty
	\]
	uniformly in $\Gamma$ and weakly in  $W^{1,q} (\Gamma)$ 
	for every $q < \infty$. 

	Moreover, any possible limit $u_\infty$ is a minimizer for 
	\[
		\Lambda_\infty(\Gamma,\V_D) = 
		\inf
		\left\{\dfrac{\| v^{\prime} \|_{L^\infty(\Gamma)}}
		{\|v \|_{L^\infty(\Gamma)}}
	 	\colon v\in W^{1,\infty}(\Gamma), v =0 \mbox{ on } \V_D,
	 	v\neq0
	 	\right\}.
	\]
	This value $\Lambda_\infty(\Gamma,\V_D)$ is the limit of 
	$\lambda_{1,p}(\Gamma,V_D)^{1/p}$ and can be characterized as 
	\[
		\Lambda_\infty(\Gamma,\V_D) 
		= \dfrac{1}{\displaystyle\max_{x \in \Gamma} d(x,\V_D)}.
	\]
\end{teo}

Note that
$$
	\max_{z \in \Gamma} d(x,\V_D)  = 
	\frac12 \max_{z \in \Gamma} \min_{\v \in \V_D} d(x,\v).
$$

\begin{proof} In this proof we use ideas from \cite{JLM2}. Let $u_p$ be an eigenfunction associated with 
$\lambda_{1,p} (\Gamma,V_D)$ 
normalized by $\|u_p\|_{L^p(\Gamma)}=1.$ 
We first prove a uniform bound (independent of $p$) for the $L^p$-norm of 
$u_p^\prime$.To this end, take $v$ any smooth function that vanishes on 
$\V_D$. Using that $u_p$ is a minimizer for \eqref{eq:autovalor} we obtain
$$
	\dfrac{\displaystyle \int_{\Gamma} |u_p^{\prime}(x)|^p\, dx}
	{\displaystyle\int_{\Gamma} |u_p (x)|^p\, dx} \leq 
	\dfrac{\displaystyle \int_{\Gamma} |v^{\prime}(x)|^p\, dx}
	{\displaystyle\int_{\Gamma} |v(x)|^p\, dx},
$$
hence we get
$$
	\left(\int_{\Gamma} |u_p^{\prime}(x)|^p\, dx \right)^{1/p} 
	\leq \left( \dfrac{\displaystyle\int_{\Gamma} |v^{\prime}(x)|^p\, dx}
	{\displaystyle\int_{\Gamma} |v(x)|^p\, dx} \right)^{1/p}.
$$
Now we observe that 
$$
	\left( \dfrac{\displaystyle\int_{\Gamma} |v^{\prime}(x)|^p\, dx}
	{\displaystyle\int_{\Gamma} |v(x)|^p\, dx} \right)^{1/p}\to 
	\dfrac{\| v^{\prime} \|_{L^\infty(\Gamma)} }{\|v \|_{L^\infty(\Gamma)} }
$$
as $p\to \infty$. Therefore, we conclude that there exists a constant $C$ 
independent of $p$ such that
$$
	\left(\int_{\Gamma} |u_p^{\prime}(x)|^p\, dx \right)^{1/p} \leq C.
$$
Then, by H\"older inequality, we have 
\begin{align*}
	\left(\int_{\Gamma} |u_p^{\prime}(x)|^q\, dx \right)^{1/q} 
	&\leq \mathrm{card}(\E(\Gamma))^{\nicefrac1q} 
	\left(\int_{\Gamma} |u_p^{\prime}(x)|^p\, dx \right)^{1/p} 
	\ell(\Gamma)^{(p-q)/pq} \\
	&\leq C  \mathrm{card}(\E(\Gamma))^{\nicefrac1q}
	\ell(\Gamma)^{\nicefrac{(p-q)}{pq}} 
\end{align*}
for all $1\le q\le p.$ Then we obtain that the family 
$\{ u_p \}_{p\ge q}$ is bounded in $W^{1,q} (\Gamma)$ for any 
$q < \infty$ and therefore by a diagonal 
procedure we can extract a sequence 
$p_j \to \infty$ such that
$$
	u_{p_j} \to u_\infty
$$
uniformly in $\Gamma$ and weakly in  $W^{1,q}(\Gamma)$ 
for every $q < \infty$. 

From our previous computations we obtain
$$
	\left(\int_{\Gamma} |u_\infty^{\prime}(x)|^q\, dx \right)^{1/q} 
	\leq \dfrac{\| v^{\prime} \|_{L^\infty(\Gamma)}}
	{\|v \|_{L^\infty(\Gamma)}}
	\mathrm{card}(\E(\Gamma))^{\nicefrac1q}\ell(\Gamma)^{\nicefrac{1}{q}} 
$$
and then (taking $q\to \infty$) we conclude that 
$$
	\| u_\infty^{\prime} \|_{L^\infty(\Gamma)}  \leq 
	\dfrac{\| v^{\prime} \|_{L^\infty(\Gamma)}}{\|v \|_{L^\infty(\Gamma)} },
$$
for every $v$ smooth that vanishes on $\V_D$.

Now, using that $u_{p_j}$ converges uniformly to $u_\infty$ we obtain that
$$
	\| u_\infty \|_{L^\infty(\Gamma)} = 1.
$$
In fact, we have
$$
	\begin{array}{rl}
		\displaystyle 
		\left(\int_{\Gamma} |u_\infty (x)|^p\, dx \right)^{1/p}  
		& \displaystyle \leq
\left(\int_{\Gamma} |u_\infty (x) - u_p (x)|^p\, dx \right)^{1/p} + \left(\int_{\Gamma} |u_p (x)|^p\, dx \right)^{1/p}
\\
& \displaystyle = \left(\int_{\Gamma} |u_\infty (x) - u_p (x)|^p\, dx \right)^{1/p}  + 1. 
\end{array}
$$
Now we have that 
$$
	\left(\int_{\Gamma} |u_\infty (x) - u_p (x)|^p\, dx \right)^{1/p} 
	\leq \| u_\infty - u_p \|_{L^\infty(\Gamma)} \ell(\Gamma)^{\nicefrac1p} 
	\to 0 
$$
as $p \to \infty$ and we conclude that
$
	\| u_\infty  \|_{L^\infty(\Gamma)} \leq 1$.
On the other hand,
$$
\begin{array}{l}
\displaystyle 1= \left(\int_{\Gamma} |u_p (x)|^p\, dx \right)^{1/p}  \\[10pt]
\displaystyle \quad \leq
\left(\int_{\Gamma} |u_\infty (x) - u_p (x)|^p\, dx \right)^{1/p} + \left(\int_{\Gamma} |u_\infty (x)|^p\, dx \right)^{1/p}
\end{array}
$$
and then we obtain the reverse inequality,
$
\| u_\infty  \|_{L^\infty(\Gamma)} \geq 1$.

We have proved that $u_\infty$ is a minimizer for
\[
	\Lambda_\infty(\Gamma,\V_D) = 
	\inf
	\left\{\dfrac{\| v^{\prime} \|_{L^\infty(\Gamma)}}
	{\|v \|_{L^\infty(\Gamma)}}
	 \colon v\in W^{1,\infty}(\Gamma), v =0 \mbox{ on } \V_D,
	 v\neq0
	 \right\}.
\]
and that
$$
	\lambda_{1,p}(\Gamma,\V_D)^{1/p} \to \Lambda_\infty(\Gamma,\V_D)
$$
as $p\to \infty$.

It remains to show that
$$
	\Lambda_\infty(\Gamma,\V_D) = \dfrac{1}{\displaystyle
	\max_{z \in \Gamma} d(z,\V_D) }.
$$
To this end, first let us consider a point $z_0 \in \Gamma$ such that
$$
	\max_{z \in \Gamma} d(z,V_1) = d(z_0, V_1)
$$
and the cone
$$
	v(x) = \left( 1-\frac{1}{d(z_0, \V_D)} d(x, z_0)  \right)_+.
$$
This function $v$ is Lipschitz and vanishes on $V_D$, 
hence it is a competitor for the infimum for $\Lambda_\infty(\Gamma,\V_D)$ and then we get
$$
	\Lambda_\infty(\Gamma,\V_D) \leq \frac{1}{d(z_0, \V_D)} = 
	\dfrac{1}{\max_{z \in \Gamma} d(z,\V_D) }.
$$
To see the reverse inequality we argue as follows: let $v$ be a smooth function 
vanishing on $\V_D$ and normalize it according to 
$\| v \|_{L^\infty(\Gamma)} =1$. 
Let $z_1 \in \Gamma$ be such that $v(z_1)=1$. 
Since $z_1 \in \Gamma$ it holds that
$$
	\max_{z \in \Gamma} d(z,\V_D) \geq d(z_1,\V_D).
$$
Hence there is a vertex $\v\in V_D$ such that 
$$
	\max_{z \in \Gamma} d(z,\V_D) \geq d(z_1,\v),
$$
and we get
$$
	1= v(z_1) - v(\v) = v' (\xi) d(z_1,\v) 
	\leq |v^{\prime} (\xi)|\max_{z \in \Gamma} d(z,\V_D).
$$
We conclude that
$$
	\| v^{\prime} \|_{L^\infty(\Gamma)} \geq 
	\dfrac{1}{\max_{z \in \Gamma} d(z,\V_D) }
$$
and therefore
$$
	\Lambda_\infty(\Gamma,\V_D) 
	\geq  \dfrac{1}{\max_{z \in \Gamma} d(z,\V_D) }.
$$
This ends the proof. \qed
\end{proof}

\section{The limit as $p \to 1$. } \label{sect-1}

In this section we study the other limit case, $p=1$. We will use functions of bounded variation
on the graph (that we will denote by $BV(\Gamma)$) and the perimeter of a subset of the graph (denoted by $\mbox{Per} (D)$). We refer to \cite{ambrosio} for precise definitions and properties of functions and sets in this context. 

\begin{teo} \label{teo-uno} 
	Let $\Gamma$ be a connected compact metric graph,
	$\V_D$ be a non-empty subset of $\V(\Gamma),$ and 
	$u_p$ be a minimizer for \eqref{eq:autovalor} normalized by 
	$\|u_p\|_{L^1(\Gamma)}=1.$
	Then, there exists a sequence $p_j \to 1^+$ such that 
	$$
		u_{p_j} \to u_1
	$$
	in $L^1 (\Gamma).$ 

	Moreover, any possible limit $u_1$ is a minimizer for 
	$$
		\Lambda_1(\Gamma,V_D) = 
		\inf\left\{\dfrac{\| v^{\prime} \|_{BV(\Gamma)}}
		{\|v \|_{L^1(\Gamma)} }\colon v\in BV(\Gamma), v=0 \mbox{ on } 
		\V_D, v\neq0
		\right\}.
	$$
	This value $\Lambda_1(\Gamma,V_D)$ is the limit of 
	$\lambda_{1,p}(\Gamma,V_D)$.
\end{teo}

\begin{proof} 
	Without loss of generality, we can assume that $u_p(x)\ge0$ for all 
	$x\in\Gamma.$ Let $v_p=(u_p)^p.$  Then 
	$v_p\in W^{1,1}(\Omega)$ and
	\begin{align*}
		\int_{\Gamma} |v_p(x)|\, dx &=1\\
		\int_{\Gamma} |v_p^\prime(x)|\, dx &= p 
		\int_{\Gamma} u(x)^{p-1}|u^\prime(x)|\, dx\\
		&\le 
		p\left(\int_{\Gamma} u(x)^p \, dx\right)^{\nicefrac1{p^\prime}}
		\left(\int_{\Gamma} |u^\prime(x)|^p \, dx\right)^{\nicefrac1{p}}\\
		&=p\left(\int_{\Gamma} |u^\prime(x)|^p \, 
		dx\right)^{\nicefrac1{p}} .
	\end{align*} 
	Hence
		\begin{align*}
			\Lambda_1(\Gamma,V_D) 
			&\leq \dfrac{\| v_p^{\prime} \|_{BV(\Gamma)}}
			{\|v_p \|_{L^1(\Gamma)}}\leq p\dfrac{\left(\displaystyle \int_\Gamma |u_p^{\prime}(x)|^p
		\, dx\right)^{\nicefrac{1}{p}}}
		{\displaystyle\left(\int_\Gamma |u_p(x) |^p\, dx\right)^{
		\nicefrac{1}{p}} } \\ 
		&=p\lambda_{1,p}(\Gamma,\V_D)^{\nicefrac1p}.
	\end{align*}
		From where we get
	\begin{equation} \label{ineq.1}
		\Lambda_1(\Gamma,\V_D) \leq 
		\liminf_{p\to 1^+} \lambda_1(\Gamma,\V_D)^{\nicefrac1p}.
	\end{equation}

	On the other hand, for any smooth function $v$ that vanishes on 
	$\V_D$ we have 
	$$
		\lambda_1(\Gamma,\V_D)^{1/p} \leq 
		\dfrac{\left(\displaystyle\int_\Gamma |v^{\prime}(x)|^p\, dx\right
		)^{1/p} }{\left(\displaystyle\int_\Gamma 
		|v(x) |^p\, dx \right)^{1/p} } 
	$$
	from where it follows
	$$
		\limsup_{p\to 1^+} \lambda_1(\Gamma,\V_D)^{1/p} \leq 
		\dfrac{\displaystyle\int_\Gamma |v^{\prime}(x)|\, dx}
		{\displaystyle\int_\Gamma |v(x) |\, dx  }
	$$
	and we conclude that 
	\begin{equation}\label{ineq.2}
		\limsup_{p\to 1^+} \lambda_1(\Gamma,\V_D)^{\nicefrac1{p}} 
		\leq \Lambda_1(\Gamma,\V_D).
	\end{equation}
	Therefore, from \eqref{ineq.1} and \eqref{ineq.2} we obtain
	\begin{equation}\label{eq}
		\lim_{p\to 1^+} \lambda_1(\Gamma,\V_D) = \Lambda_1(\Gamma,\V_D).
	\end{equation}

	Moreover, by \cite[Theorem 4 Section 5.2.3]{EG} we have that there is 
	$u_1 \in BV (\Gamma)$ such that
	$$
		\| u_{p_j} - u_1\|_{L^1 (\Gamma)} \to 0
	$$
	for a sequence $p_j \to 1^+$. 
	From the lower semicontinuity of the variation measure (see 
	\cite[Theorem 1 Section 5.2.1]{EG}), we have
	$$
		\| u_1 \|_{BV (\Gamma) } \leq \liminf_{{p_j}\to 1}  
		\| u_{p_j} \|_{BV (\Gamma)} .
	$$
	From this we conclude that every possible limit of a sequence of $u_p$ as $p\to 1$ is an 
	extramal for $\Lambda_1(\Gamma,\V_D)$.\qed
	\end{proof}

\begin{teo} \label{teo-uno-cheeger} It holds that 
	$$
		\Lambda_1(\Gamma,\V_D) = 
		\inf\left\{ \dfrac{\mathrm{Per}(D)}{|D |} \colon
		D\subset\Gamma, \, D\cap \V_D=
		\emptyset\right\}.
	$$ 
	\end{teo}

	\begin{proof}
		We have 
		$$
			\Lambda_1(\Gamma,\V_D)  \leq \lambda=  
			\inf\left\{ \dfrac{\mathrm{Per}(D)}{|D |} \colon
			D\subset\Gamma, \, D\cap \V_D=
			\emptyset\right\} . 
		$$
	By Theorem \ref{teo-uno} there exists a function 
	$u \in BV (\Gamma)$, $u\neq 0$, such that
	$$
		\Lambda_1(\Gamma,\V_D)   
		= \dfrac{\| u^{\prime} \|_{BV(\Gamma)}}{\|u \|_{L^1(\Gamma)} }.
	$$
	We can consider without loss of generality that $u \geq 0$. 
	Let 
	$$
		E_t \coloneqq \{x \in \Gamma \colon  u(x) > t\}.
	$$
	We have
	$$
		|u^{\prime}| (\Gamma) = \int_0^{\infty} \mathrm{Per} (E_t)   dt.
	$$
	Hence, we get using Cavalieri's principle, 
	\begin{align*}
		0 &= \| u^{\prime} \|_{BV(\Gamma)} - 
		\Lambda_1(\Gamma,\V_D) \|u \|_{L^1(\Gamma)}\\  
		&=  \int_0^{\infty} ( \mathrm{Per}(E_t) - 
		\Lambda_1(\Gamma,\V_D)  |E_t|)  dt \\
		&\geq \int_0^{\infty} (\mathrm{Per}(E_t) - \lambda  |E_t|)  
		dt \geq 0.
	\end{align*}
	
	Therefore, we conclude that
	for almost every $t \in \mathbb{R}$ 
	(in the sense of the Lebesgue measure on $\mathbb{R}$),
	$$
		\mathrm{Per}(E_t) = \lambda |E_t|
	$$
	and 
	$$
		\lambda = \Lambda_1(\Gamma,\V_D).
	$$\qed
\end{proof}

Sets $D^*$ such that
$$
 	\inf\left\{ \dfrac{\mathrm{Per}(D)}{|D |} \colon
		D\subset\Gamma, \, D\cap \V_D=
		\emptyset\right\} 
		 = \dfrac{\mathrm{Per} (D^*)}{|D^*|} 
$$
are called Cheeger sets. See \cite{parini} and references therein.

\begin{ex}
	 To see that the optimal value $\Lambda_1(\Gamma,\V_D)$ 
	 depends strongly on the geometric configuration of the graph 
	 $\Gamma$, let us consider the following example: let 
	 $\Gamma$ be a simple graph with 4 nodes 
	 (3 of them, the terminal nodes, are in $\V_D$) and 
	 3 edges as the one described by the next figure:
	 
	 
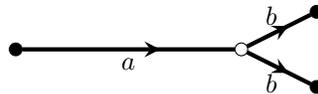
\begin{center}
		\begin{tikzpicture}
  			\node  [fill, circle,draw, scale=.5]  (a) at (0,0) {};
  			\node  (a1) at (1.5,-0.2)  { $a$ };
  			\node  [ circle,draw, scale=.5]  (b) at (3,0)  {};
  			\node  [fill, circle,draw, scale=.5]  (c) at (4,-.5) {};
  			\node  (b1) at (3.4,-0.45)  { $b$ };
  			\node  [fill, circle,draw, scale=.5]  (d) at (4,.5) {};
  			\node  (b1) at (3.4,0.45)  { $b$ };
  			\draw[directed,ultra thick] (a) -- (b);
  			\draw[directed,ultra thick] (b) -- (c);
  			\draw[directed,ultra thick]  (b) -- (d);
		\end{tikzpicture}
	\end{center}
	Let us compute 
	\begin{align*}
		\Lambda_1(\Gamma,\V_D)&=
		\inf\left\{\dfrac{\| v^{\prime} \|_{BV(\Gamma)}}
		{\|v \|_{L^1(\Gamma)} }\colon v\in BV(\Gamma), v=0 \mbox{ on } 
		\V_D, v\neq0
		\right\}\\
		&=\inf\left\{ \dfrac{\mathrm{Per}(D)}{|D |} \colon
		D\subset\Gamma, \, D\cap \V_D=
		\emptyset\right\}, 
	\end{align*}
	in this case. As we will see its value 
	(and the corresponding optimal set $D^*$) depends 
	on the lengths $a$ and $b$.

	First, let us compute the value of 
	$\dfrac{\mathrm{Per}(D)}{|D |}$ for $D=\Gamma$. 
	We have
	$$
		|\Gamma|=\ell(\Gamma) = a+2b, \qquad \mbox{and} 
		\qquad \mathrm{Per}(\Gamma) =3.
	$$
	Hence
	$$
	\dfrac{ \mathrm{Per}(\Gamma)}{\ell(\Gamma)} = \dfrac{3}{a+2b}.
	$$

	On the other hand, if we consider $D_a$ 
	the characteristic function of the edge of length $a$ 
	we obtain
	$$
		|D_a| = a, \qquad \mbox{and} \qquad \mathrm{Per}(D_a) =2,
	$$
	and then
	$$
		\dfrac{ \mathrm{Per}(D_a)}{|D_a|} = \dfrac{2}{a}.
	$$

	Now we remark that any other subset $D$ of $\Gamma$ has a ratio 
	$\dfrac{ \mathrm{Per}(D)}{|D|}$ bigger or equal than one of 
	the previous two sets. Therefore,
	we conclude that
	$$
		\Lambda_1(\Gamma,\V_D) =
		\begin{cases}
			\dfrac{3}{a+2b}, &\text{ if }a \leq 4b,\\[10pt]
			\dfrac{2}{a} & \mbox{if } a > 4b.
		\end{cases}
	$$
\end{ex}

\bigskip

{\bf Acknowledgements.} We want to thank Carolina A. Mosquera for her encouragement and several interesting discussions.

\bibliographystyle{amsplain}

\begin{thebibliography}{1}

\bibitem{ambrosio} L. Ambrosio, N. Fusco, D. Pallara, Functions of bounded variations and free discontinuity problems, Oxford University Press, 2000.

\bibitem{anane} A. Anane, {\it Simplicite et isolation de la premiere valeur propre du p-laplacien avec poids}, C. R. Acad. Sci. Paris Ser. I Math. 305 (1987), no. 16, 725?728 (French).

\bibitem{Liviu} V. Banica and L. I. Ignat. {\it Dispersion for the Schrödinger equation on networks}. J. Math. Phys. 52 (2011), no. 8, 083703, 14 pp. 


\bibitem{BK} G. Berkolaiko and P. Kuchment, 
	{\em Introduction to quantum graphs}. 
	Mathematical Surveys and Monographs, 186. 
	American Mathematical Society, Providence, RI, 2013. xiv+270 pp.
	
	
\bibitem{BD1} I. Birindelli and F. Demengel, {\it First eigenvalue and maximum principle for fully nonlinear singular operators}, Adv. Differential Equations 11, (2006), no. 1, 91--119.


\bibitem{BD2} I. Birindelli and F. Demengel, {\it Eigenvalue, maximum principle and regularity for fully non linear homogeneous operators}, Commun. Pure
Appl. Anal. 6, (2007), no. 2, 335--366.

\bibitem{bon}  A. N. Bondarenko and V. A. Dedok. {\it Spectral Surgery for the SchrÃ¶dinger Operator on Graphs}.
Doklady Mathematics, 2012, Vol. 85, No. 3, 367--368.

\bibitem{CCN} V. Caselles, A. Chambolle, M. Novaga, {\it Some remarks on uniqueness and regularity of Cheeger sets}, Rendiconti del Seminario Matematico della Universit`a di Padova 123 (2010), 191--201.

\bibitem{Euler}  L. Euler, 
	{\em Solutio problematis ad geometriam situs pertinentis}.
	Comment. Academiae Sci. I. Petropolitanae 8 (1736), 128--140.
	
\bibitem{EG} L. Evans and R. Gariepy, 
{\it Measure theory and fine properties of functions}.
Studies in Advanced Mathematics. CRC Press, Boca Raton, FL,  1992. viii+268 pp.


\bibitem{FMP} A. Figalli, F. Maggi, A. Pratelli, {\it A note on Cheeger sets}, Proceedings of the American Mathematical Society 137 (2009), 2057--2062.


\bibitem{cita} L. Friedlander, {\it Genericity of simple eigenvalues for a metric graph}. Israel J. Math. 146 (2005), 149--156. 


\bibitem{GP1} J. Garcia-Azorero and I. Peral, {\it
Existence and non-uniqueness for the $p-$Laplacian: nonlinear eigenvalues}.
Comm. Partial Differential Equations. Vol. 12 (1987), 1389-1430.

\bibitem{jorge-pepe} J. Garcia Melian, J. Sabina de Lis, {\it On the perturbation of eigenvalues for the 
$p-$Laplacian}, Comptes Rendus Acad. Sci. Ser. I Math. 332 (10) (2001), 893-898.


\bibitem{HC} C.Hierholzer and  C. Wiener, 
	{\em Ueber die Maglichkeit, einen Linienzug ohne Wiederholung 
	und ohne Unterbrechung zu umfahren}. (German) 
	Math. Ann. 6 (1873), no. 1, 30--32. 
	
	
\bibitem{Juu} P. Juutinen, {\it Principal eigenvalue of a very
badly degenerate operator and applications}. J. Differential
Equations, 236 (2007), 532--550.


 \bibitem{JLM2} P. Juutinen, P. Lindqvist and J.J. Manfredi, {\it The $\infty$-eigenvalue
 problem.} Arch. Ration. Mech. Anal. { 148} (1999), 89--105.
 
 \bibitem{cita2} V. Kostrykin, R. Schrader, {\it Kirchhoff?s Rule for Quantum Wires}.
J. Phys. A 32 (1999), no. 4, 595--630.

\bibitem{K} P. Kuchment, {\it Quantum graphs: I. Some basic structures}, 
Waves Random Media 14 (2004), S107--S128.

\bibitem{Ku} P. Kurasov, {\it On the Spectral Gap for Laplacians on Metric Graphs}.
Acta Physica Polonica A. 124 (2013), 1060--1062.

\bibitem{KN}  P. Kurasov and S. Naboko, 
	{\em On Rayleigh theorem for quantum graphs}. 
	Institut Mittag-Leffler Report No. 4, 2012/2013.
	
	\bibitem{KN2} P. Kurasov and S. Naboko, {\it  
Rayleigh estimates for differential operators on graphs}. 
J. Spectr. Theory 4 (2014), no. 2, 211--219.

\bibitem{Ku22} P. Kurasov, G. Malenova and S. Naboko, {\it Spectral gap for quantum graphs and their edge connectivity}. J. Phys. A 46 (2013), no. 27, 275309, 16 pp.

\bibitem{LE} J. Lang and D. Edmunds,
	{\em Eigenvalues, embeddings and generalised trigonometric 
	functions}. Lecture Notes in Mathematics, 2016. 
	Springer, Heidelberg, 2011. xii+220 pp.
	
	\bibitem{11} P. Lindqvist, {\it Note on a nonlinear eigenvalue problem}, Rocky Mountain J. Math. 23 (1993), 281--288.
	
	\bibitem{12} P. Lindqvist, {\it Some remarkable sine and cosine functions}, Ricerche di Matematica XLIV (1995), 269--290.
	
	\bibitem{13} P. Lindqvist and J. Peetre, {\it Two remarkable identities, called Twos, for inverses to some Abelian integrals}, Amer. Math. Monthly 108 (2001), 403--410.
	
	
	\bibitem{parini} E. 
	Parini, {\it An introduction to the Cheeger problem}. Surveys in Mathematics and its Applications. Vol. 6 (2011), 9--22.
	
	\bibitem{olaf} O. Post, {\it Spectral Analysis
on Graph-Like Spaces}. Lecture Notes in Mathematics. 2012.
	
	\bibitem{vazquez} J. V\'azquez, 
 {\it A strong maximum principle for some quasilinear elliptic equations},
Appl. Math. Optim. 12 (1984), no. 3, 191--202.


\end{thebibliography}

\end{document}